\documentclass[12pt,draftclsnofoot,hidelinks,onecolumn]{IEEEtran}
\usepackage{amsmath,amsfonts,amsthm}
\usepackage{amssymb}
\usepackage{algorithmic}
\usepackage{algorithm}
\usepackage[caption=false,font=normalsize,labelfont=sf,textfont=sf]{subfig}
\usepackage{textcomp}
\usepackage{stfloats}
\usepackage{url}
\usepackage{verbatim}
\usepackage{graphicx}
\usepackage{cite}
\usepackage{color}
\usepackage{diagbox}
\newtheorem{theorem}{Theorem}

\newtheorem{proposition}{Proposition}

\begin{document}

\title{Channel Estimation for RIS-Aided MIMO Systems: A Partially Decoupled Atomic Norm Minimization Approach}

\author{Yonghui~Chu, Zhiqiang~Wei,~\IEEEmembership{Member,~IEEE,} Zai~Yang,~\IEEEmembership{Senior~Member,~IEEE,} and Derrick~Wing~Kwan~Ng,~\IEEEmembership{Fellow,~IEEE}%
\thanks{This work is accepted by IEEE Transactions on Wireless Communications. A preliminary version of this work was presented in part at the IEEE Global Communications Conference (GLOBECOM), 2023 [DOI: 10.1109/GLOBECOM54140.2023.10436833]. \nocite{chu2023GLOBECOM}}%
\thanks{Yonghui~Chu, Zhiqiang~Wei, and Zai~Yang are with the School of Mathematics and Statistics, Xi’an Jiaotong University, Xi’an, Shaanxi 710049, China, and also with the Peng Cheng Laboratory, Shenzhen, Guangdong 518055, China, and also with the Pazhou Laboratory (Huangpu), Guangzhou, Guangdong 510555, China (e-mail: chuyonghui@stu.xjtu.edu.cn; zhiqiang.wei@xjtu.edu.cn; yangzai@xjtu.edu.cn).}%
\thanks{Derrick~Wing~Kwan~Ng is with the School of Electrical Engineering and Telecommunications, University of New South Wales, Sydney, NSW 2052, Australia (e-mail: w.k.ng@unsw.edu.au).}}

%{Shell \MakeLowercase{\emph{et al.}}: A Sample Article Using IEEEtran.cls for IEEE Journals}


\maketitle

\begin{abstract}
Channel estimation (CE) plays a key role in reconfigurable intelligent surface (RIS)-aided multiple-input multiple-output (MIMO) communication systems, while it poses a challenging task due to the passive nature of RIS and the cascaded channel structures. In this paper, a partially decoupled atomic norm minimization (PDANM) framework is proposed for CE of RIS-aided MIMO systems, which exploits the three-dimensional angular sparsity of the channel. In particular, PDANM partially decouples the differential angles at the RIS from other angles at the base station and user equipment, reducing the computational complexity compared with existing methods. A reweighted PDANM (RPDANM) algorithm is proposed to further improve CE accuracy, which iteratively refines CE through a specifically designed reweighing strategy. Building upon RPDANM, we propose an iterative approach named RPDANM with adaptive phase control (RPDANM-APC), which adaptively adjusts the RIS phases based on previously estimated channel parameters to facilitate CE, achieving superior CE accuracy while reducing training overhead. Numerical simulations demonstrate the superiority of our proposed approaches in terms of running time, CE accuracy, and training overhead. In particular, the RPDANM-APC approach can achieve higher CE accuracy than existing methods within less than 30 percent training overhead while reducing the running time by tens of times.
\end{abstract}

\begin{IEEEkeywords}
Atomic norm minimization, adaptive phase control, channel estimation, multiple-input multiple-output, reconfigurable intelligent surface.
\end{IEEEkeywords}

\section{Introduction} \label{SecInt}

\IEEEPARstart{R}{econfigurable} intelligent surface (RIS), also known as intelligent reflecting surface (IRS) \cite{wu2019towards}, is a metasurface equipped with integrated circuits. It is typically composed of passive reflecting elements whose phases can be adapted independently to the instantaneous channel state information by an intelligent controller \cite{yuan2021reconfigurable}. By adapting the phases of the reflecting elements, RIS can alter its electromagnetic response to the incident signals, thus compensating for the severe propagation path loss or mitigating potential interference \cite{liu2021reconfigurable}. Since RIS only employs passive reflecting elements without the use of active radio-frequency (RF) chains, it is generally more energy-efficient and cost-effective than traditional active relays \cite{huang2019reconfigurable,di2020reconfigurable}. In addition, the compact size of RIS components facilitates their mass deployment in various structures, such as billboards, windows, and building facades, etc. \cite{liu2021reconfigurable}. Due to its low cost, low power consumption, high flexibility, and the ability to reconfigure wireless
propagation environments, RIS is considered a promising technology for next-generation wireless networks \cite{yuan2021reconfigurable,liu2021reconfigurable,huang2019reconfigurable,di2020reconfigurable,basar2019wireless,wei2021sum,hu2021robust,wu2021intelligent}.

To fully unleash the potential performance gains brought by RIS, channel estimation (CE) is crucial in RIS-aided communication systems but challenging in practice \cite{wu2021intelligent,yuan2021reconfigurable}. Firstly, in RIS-aided communication systems, apart from the direct link between the base station (BS) and the user equipment (UE), the BS-to-RIS channel and the RIS-to-UE channel need to be estimated. Yet, the cascaded channel structure poses a challenge for effective CE. Secondly, since the passive reflecting elements of RIS are generally equipped with phase shifters only \cite{yuan2021reconfigurable,liu2021reconfigurable}, the signal processing capability of RIS is limited, making it impossible to transmit or receive training signals for CE. In addition, since a RIS typically consists of a large amount of elements, the number of channel parameters to be estimated increases substantially in RIS-aided systems, imposing extra challenges and leading to overwhelming training overhead.

In fact, CE for RIS-aided communication systems has been widely studied in the literature, e.g., \cite{jensen2020optimal,wei2020parallel,de2021channel,wang2020compressed,ardah2021trice,he2021channel,he2021leveraging,chung2021location}. In particular, various classical CE methods, such as least squares (LS), were applied to RIS-aided communication systems in the early stage. For instance, a CE scheme employing LS based on the minimum variance unbiased estimation criterion was designed in \cite{jensen2020optimal}. By exploiting the parallel factorization of the cascaded BS-RIS-UE channel, an iterative CE algorithm for multiple-input single-output (MISO) systems that capitalized alternating LS was proposed in \cite{wei2020parallel} and further generalized to multiple-input multiple-output (MIMO) systems in \cite{de2021channel}. In addition, another parallel factorization-based method called Khatri-Rao factorization (KRF) was proposed in \cite{de2021channel}, which has a lower complexity since it does not require iteration. However, the aforementioned CE methods generally require a quantity of training overhead that depends on the size of the RIS, since they do not sufficiently exploit the channel structure.

To reduce the required training overhead, compressed sensing (CS) technology has been applied to the CE for RIS-aided communication systems, which exploits the angular sparsity of the channel. For instance, in \cite{wang2020compressed}, CE was formulated as a sparse signal recovery problem and was solved by classical CS methods, such as orthogonal matching pursuit (OMP) \cite{pati1993orthogonal}. For RIS-aided MIMO systems, a two-stage RIS-aided channel estimation (TRICE) method was proposed in \cite{ardah2021trice}, where the CE problem is decomposed into two subproblems with each formulated as an angle-of-arrival (AoA) estimation problem and solved by the OMP method. However, since existing CS-based methods assume that the channel angular parameters lie in a set of grid points, they suffer from the grid mismatch problem \cite{stoica2012sparse,yang2018sparse} caused by the limited grid resolution, which leads to a low CE accuracy.

Most recently, CE methods based on atomic norm minimization (ANM) \cite{yang2018sparse} have been developed for RIS-aided communication systems \cite{he2021channel,he2021leveraging,chung2021location}, which directly estimate the desired channel angular parameters in the continuous domain and thus improve the CE accuracy.  Specifically, a two-stage CE method was proposed in \cite{he2021channel}, where the channel parameters were divided into two groups and estimated by the ANM in two stages, respectively. Unfortunately, a large amount of training overhead is required in the first stage to reduce the possible error propagating to the second stage. When the location information of the RIS and BS is available, the authors in \cite{he2021leveraging} proposed to leverage the location information for enhancing the CE performance and for saving the training overhead. Besides, in \cite{chung2021location}, two one-stage ANM-based CE methods were proposed for RIS-aided MIMO communication systems, i.e., ANM with two-dimensional (2D) atoms (ANM-2D) and ANM with three-dimensional (3D) atoms (ANM-3D). Although ANM-3D is generally superior to ANM-2D due to its more accurate characterization of the angular structure of the channel, a large-scale semidefinite programming (SDP) problem is required to be solved in ANM-3D, which has an order of magnitude higher computational complexity than that of the ANM-2D. This motivates us to develop novel ANM-based CE approaches to strike a balance between the computational complexity and CE accuracy.

In this paper, we first propose a partially decoupled ANM (PDANM) framework with reduced computational complexity compared to state-of-the-art approaches. Based on PDANM, an iterative algorithm called reweighted PDANM (RPDANM) that achieves improved CE accuracy is proposed. To further reduce the training overhead, we propose a RPDANM with adaptive phase control (RPDANM-APC) approach that realizes high-accuracy CE with low training overhead for RIS-aided MIMO systems. The contributions of the paper are summarized as follows:
\begin{itemize}
	\item{We propose a one-stage PDANM framework for CE of RIS-aided MIMO systems, which exploits the 3D angular structure of the considered channel in a partially decoupled manner. Firstly, we derive an effective channel model, which not only resolves the inherent parameter ambiguity in the previous cascaded channel model, but also exhibits a simpler structure that facilitate the CE. On this basis, we define the partially decoupled atomic norm (PDAN) and reformulate the considered CE problem as a PDANM problem, which decouples the 3D angular structure of the effective channel into two lower-dimensional angular structures. Moreover, we formulate an SDP problem to efficiently solve the PDANM problem and establish the conditional equivalence between them.}
	\item{We propose an iterative CE algorithm called RPDANM, which promotes the 3D angular sparsity over PDANM by approximately solving a rank minimization problem instead of its convex relaxation. Compared with PDANM, RPDANM enhances CE accuracy through a specifically designed atom reweighting strategy without introducing additional training overhead. In particular, each iteration of RPDANM can be regarded as solving a weighted PDANM (WPDANM) problem with adaptively updated weighting functions.}
	\item{To reduce training overhead and to facilitate CE, we consider the design of the RIS phase control matrix during the training phase. By adaptively adjusting the phases of the RIS during the channel sounding procedure, we propose a RPDANM-APC approach that achieves promising CE accuracy with a limited training overhead.}
\end{itemize}

The rest of the paper is organized as follows. Section \ref{SecSys} presents the system model for the considered RIS-aided MIMO communication system and introduces the adopted channel sounding procedure and the channel model. Section \ref{SecPDANM} proposes an effective channel model and a PDANM framework with a computational complexity analysis. Section \ref{SecRPDANM} proposes a PDANM-based iterative algorithm called RPDANM. Section \ref{SecRPDANMAPC} proposes a RPDANM-based iterative CE approach named RPDANM-APC. Section \ref{SecNum} provides numerical simulations to demonstrate the advantages of our proposed methods in terms of CE accuracy, running time, and training overhead. Conclusions are drawn in Section \ref{SecCon}.

\begin{table*}[t]
	\scriptsize
	\caption{Key Notations for System Model} \label{NotationMeaning}
	\begin{center}
		\begin{tabular}{ c | c | c | c }
			\hline			
			Notation & Physical Meaning & Notation & Physical Meaning\\ \hline
			$N_\mathrm{B}$ & Number of antennas at BS & $\boldsymbol{\rho}_\mathrm{BR} \in \mathbb{C}^{L_\mathrm{BR}}$ & Path gains of BS-to-RIS channel \\
			$N_\mathrm{U}$ & Number of antennas at UE & $\boldsymbol{\theta}_\mathrm{B} \in [0,\pi]^{L_\mathrm{BR}}$ & AoDs at BS \\    
			$N_\mathrm{R}$ & Number of elements of RIS & $\boldsymbol{\phi}_\mathrm{R} \in [0,\pi]^{L_\mathrm{BR}}$ & AoAs at RIS \\
			$M$ & Length of a training sequence & $\mathbf{H}_\mathrm{BR} \in \mathbb{C}^{N_\mathrm{R} \times N_\mathrm{B}}$ & BS-to-RIS channel \\
			$B$ & Number of training slots & $\boldsymbol{\rho}_\mathrm{RU} \in \mathbb{C}^{ L_\mathrm{RU}}$ & Path gains of RIS-to-UE channel \\
			$L_\mathrm{BR}$ & Path number of BS-to-RIS channel & $\boldsymbol{\theta}_\mathrm{R} \in [0,\pi]^{L_\mathrm{RU}}$ & AoDs at RIS \\
			$L_\mathrm{RU}$ & Path number of RIS-to-UE channel & $\boldsymbol{\phi}_\mathrm{U} \in [0,\pi]^{L_\mathrm{RU}}$ & AoAs at UE \\
			$\mathbf{S}_{b} \in \mathbb{C}^{N_\mathrm{B} \times M}$ & Training matrix in the $b$-th slot & $\mathbf{H}_\mathrm{RU} \in \mathbb{C}^{N_\mathrm{U} \times N_\mathrm{R}}$ & RIS-to-UE channel \\
			$\boldsymbol{\omega}_{b} \in \mathbb{C}^{N_\mathrm{R} \times 1}$ & Phase control vector for RIS in the $b$-th slot & $\boldsymbol{\rho}_\mathrm{BU} \in \mathbb{C}^{L_\mathrm{BR} L_\mathrm{RU}}$ & Effective path gains \\
			$\boldsymbol{\Omega} \in \mathbb{C}^{N_\mathrm{R} \times B}$ & Phase control matrix for RIS in all slots & $\boldsymbol{\psi}_\mathrm{R} \in [0,\pi]^{L_\mathrm{BR} L_\mathrm{RU}}$ & Differential angles at RIS \\
			$\mathbf{H}^b_\mathrm{BU} \in \mathbb{C}^{N_\mathrm{U} \times N_\mathrm{B}}$ & Cascaded channel in the $b$-th slot & $\mathbf{H} \in \mathbb{C}^{N_\mathrm{B}N_\mathrm{U} \times N_\mathrm{R}}$ & Effective channel in all slots \\
			\hline
		\end{tabular}
	\end{center}
	\vspace{-3mm}
\end{table*}

\emph{Notations:} Lowercase and uppercase bold letters represent vectors and matrices, respectively. $\mathbb{C}$ denotes the set of complex numbers. $| \cdot |$ denotes the amplitude of a scalar. $\overline{ \left( \cdot \right) }$, $\left( \cdot \right)^{T}$, $\left( \cdot \right)^{H}$ and $\left( \cdot \right)^{\dagger}$ denote the conjugate, transpose, Hermitian transpose, and Moore-Penrose pseudo inverse, respectively. $\mathrm{rank}(\mathbf{A})$, $\mathrm{tr}(\mathbf{A})$, and $\mathrm{col}(\mathbf{A})$ denote the rank, trace, and column space of matrix $\mathbf{A}$, respectively, and $\mathbf{A} \succeq \mathbf{0}$ means that $\mathbf{A}$ is a Hermitian positive semidefinite (PSD) matrix. $\mathrm{diag}(\mathbf{x})$ is a diagonal matrix with its main diagonal entries given by $\mathbf{x}$. The identity matrix of size $N$ is denoted as $\mathbf{I}_{N}$. $[\mathbf{x}]_{l}$ denotes the $l$-th entry of the vector $\mathbf{x}$, and $[\mathbf{X}]_{i,j}$ denotes the $(i,j)$-th entry of the matrix $\mathbf{X}$. $\left \| \cdot \right \|_{2}$ and $\left \| \cdot \right \|_{F}$ denote the $\ell_{2}$ norm and the Frobenius norm, respectively. $\otimes$ and $\diamond$ denote the Kronecker product and Khatri-Rao product (or column-wise Kronecker product), respectively, and $\mathrm{vec} \left( \cdot \right)$ denotes vectorization. $a \mod b$ denotes the remainder of dividing $a$ by $b$. $\mathbb{E}(\cdot)$ is the expectation operator. $\mathcal{CN}(\mu, \sigma^{2})$ denotes a complex Gaussian distribution whose mean is $\mu$ and variance is $\sigma^{2}$. For a $D$-dimensional ($D \geq 2$) tensor $\mathbf{T} = [\mathbf{T}_{1},\cdots,\mathbf{T}_{2N_{1}-1}] \in \mathbb{C}^{(2 N_{1} - 1) \times \cdots \times (2 N_{D} - 1)}$, where $\mathbf{T}_{n_{1}} \in \mathbb{C}^{1 \times (2 N_{2} - 1) \times \cdots \times (2 N_{D} - 1)}$ with $n_{1} \in \{ 1,\cdots,2N_{1}-1 \}$, a $D$-level Toeplitz matrix \cite{yang2016vandermonde} is defined recursively as
\begin{equation*}
	\hspace{-1mm} \mathcal{T}_{\mathbf{N}}(\mathbf{T}) \hspace{-1mm} = \hspace{-2mm}
	\begin{bmatrix}
		\mathcal{T}_{\mathbf{N}_{-1}}(\mathbf{T}_{N_{1}}) & \hspace{-5mm} \mathcal{T}_{\mathbf{N}_{-1}}(\mathbf{T}_{N_{1}+1}) & \hspace{-4mm} \cdots & \hspace{-4mm} \mathcal{T}_{\mathbf{N}_{-1}}(\mathbf{T}_{2N_{1}-1}) \\ \mathcal{T}_{\mathbf{N}_{-1}}(\mathbf{T}_{N_{1}-1}) & \hspace{-5mm} \mathcal{T}_{\mathbf{N}_{-1}}(\mathbf{T}_{N_{1}}) & \hspace{-4mm} \cdots & \hspace{-4mm} \mathcal{T}_{\mathbf{N}_{-1}}(\mathbf{T}_{2N_{1}-2}) \\ \vdots & \hspace{-5mm} \vdots & \hspace{-4mm} \ddots & \hspace{-4mm} \vdots \\ \mathcal{T}_{\mathbf{N}_{-1}}(\mathbf{T}_{1}) & \hspace{-5mm} \mathcal{T}_{\mathbf{N}_{-1}}(\mathbf{T}_{2}) & \hspace{-4mm} \cdots & \hspace{-4mm} \mathcal{T}_{\mathbf{N}_{-1}}(\mathbf{T}_{N_{1}}) 
	\end{bmatrix}\hspace{-1.5mm},
\end{equation*}
where $\mathbf{N} = [N_{1},\cdots,N_{D}]$ and $\mathbf{N}_{-1} = [N_{2},\cdots,N_{D}]$. In particular, a ($1$-level) Toeplitz matrix is defined as
\begin{equation*}
	\mathcal{T}_{N}(\mathbf{t}) = \hspace{-1mm}
	\begin{bmatrix}
		t_{N} & t_{N+1} & \cdots & t_{2N-1} \\ t_{N-1} & t_{N} & \cdots & t_{2N-2} \\ \vdots & \vdots & \ddots & \vdots \\ t_{1} & t_{2} & \cdots & t_{N} 
	\end{bmatrix},
\end{equation*}
where $\mathbf{t} = [t_{1},\cdots,t_{2N-1}] \in \mathbb{C}^{2N - 1}$. For clarity, we summarize the key notations for the system model adopted in this paper in Table \ref{NotationMeaning}.

\section{System Model} \label{SecSys}

\subsection{System Model}

\begin{figure}[!t]
	\centering
	\includegraphics[width=5in]{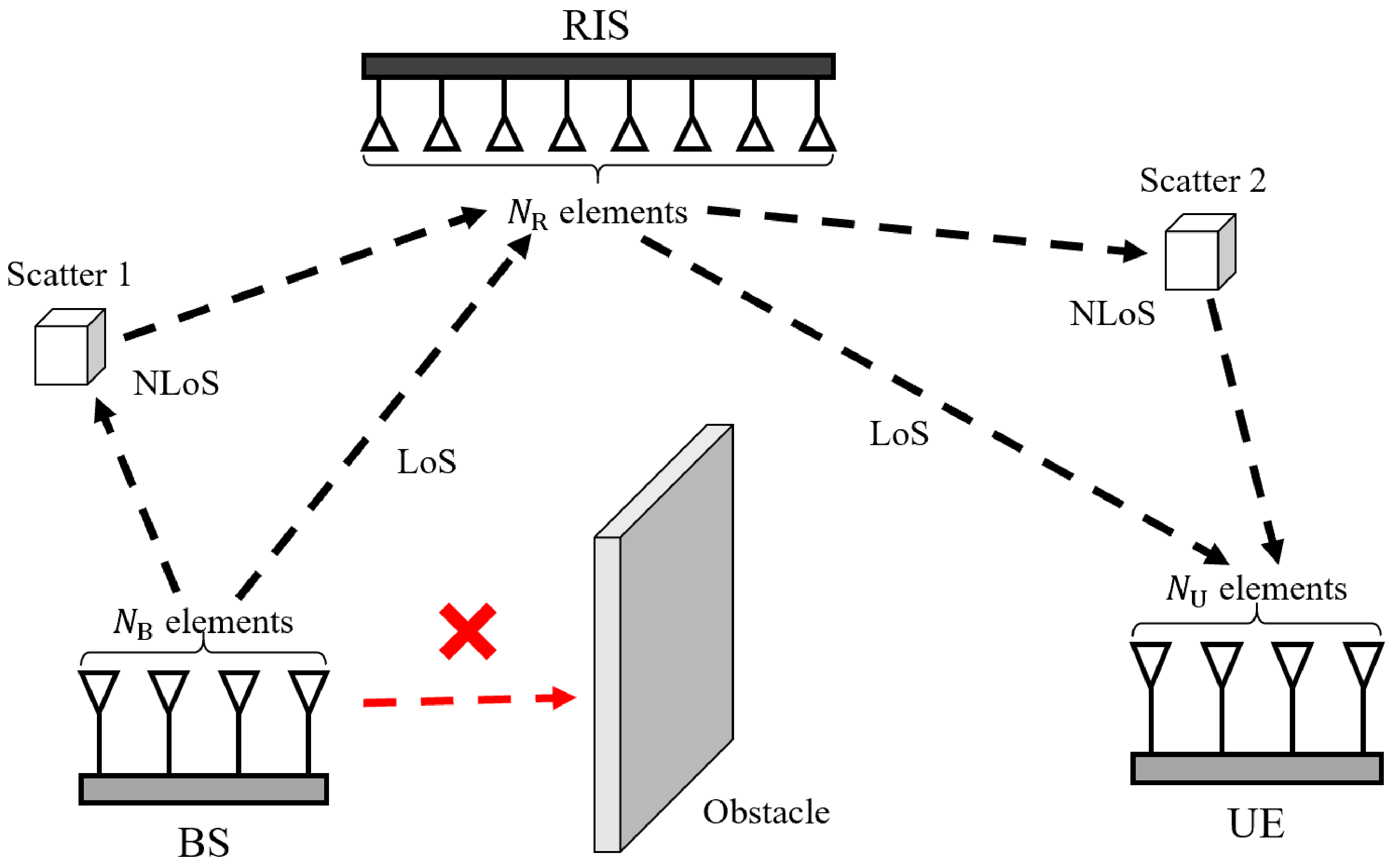}
	\caption{The considered RIS-aided MIMO communication system.}
	\label{fig1}
	\vspace{-5mm}
\end{figure}

We consider a RIS-aided MIMO communication system, where a BS communicates with a UE\footnote{In this paper, we consider a single UE in the considered RIS-aided MIMO system. The proposed CE methods can be extended to multiple UEs by employing orthogonal training sequences among different UEs, which ensures that the CE for each UE is independent of the others.} with the help of a RIS while the direct BS-to-UE channel is blocked, as shown in Fig. \ref{fig1}. The BS, UE, and RIS are equipped with uniform linear arrays (ULAs) of $N_\mathrm{B}$, $N_\mathrm{U}$, and $N_\mathrm{R}$ elements, respectively\footnote{In this paper, ULAs are employed in the considered RIS-aided MIMO system and thus only azimuth is considered to simplify the notations. The proposed CE method can be extended to the case of uniform planar arrays by considering both azimuth and elevation.}. In the considered RIS-aided MIMO system, the BS transmits training sequences, the RIS adjusts its phases to steer the signal towards the UE, and the UE estimates the cascaded BS-RIS-UE channel based on the received signals.

\begin{figure}[!t]
	\centering
	\includegraphics[width=5in]{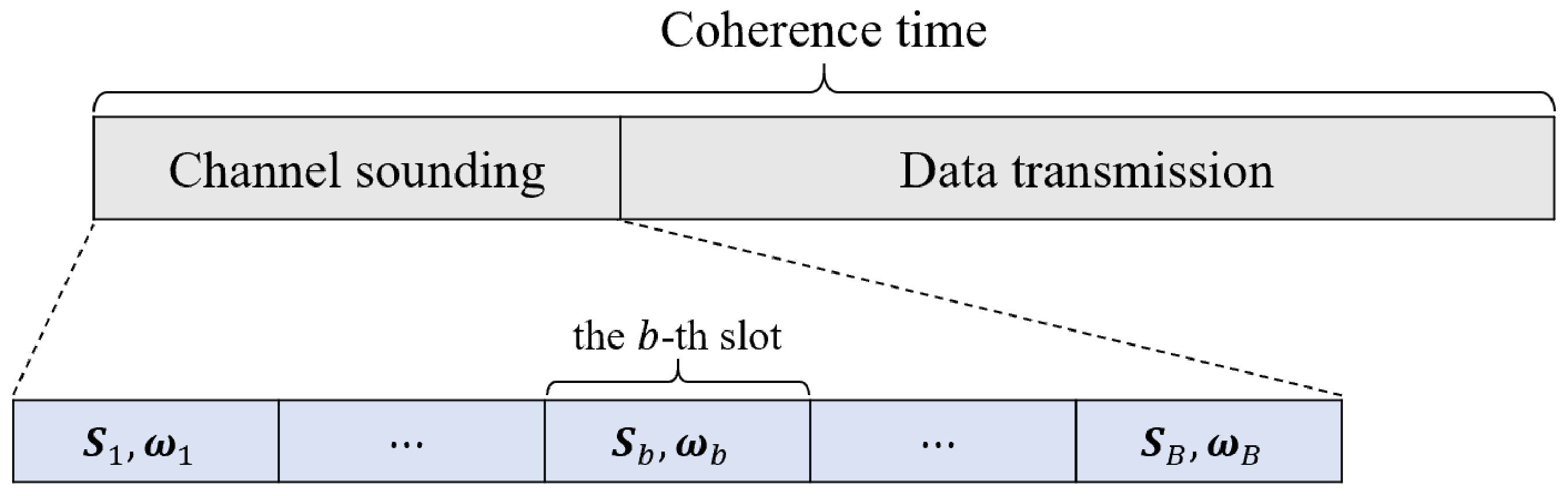}
	\caption{The adopted slot-based channel sounding procedure.}
	\label{fig2}
	\vspace{-5mm}
\end{figure}

The detailed channel sounding procedure is introduced as follows. Following \cite{he2021channel,chung2021location}, a slot-based channel sounding procedure is adopted in this paper, as illustrated in Fig. \ref{fig2}. In particular, a coherence time interval is divided into two periods for channel sounding and data transmission, respectively. The channel sounding period includes $B$ time slots and each time slot contains $M$ symbol durations. In the $b$-th time slot, $b \in \{1,\cdots,B\}$, the BS transmits $N_\mathrm{B}$ orthogonal training sequences of length $M$, collected in $\mathbf{S}_{b} = [\mathbf{s}_{b1},\cdots,\mathbf{s}_{bN_\mathrm{B}}]^{T} \in \mathbb{C}^{N_\mathrm{B} \times M}$, satisfying $M \geq N_\mathrm{B}$ and $\mathbf{S}_{b} \mathbf{S}_{b}^{H} = P \mathbf{I}_{N_\mathrm{B}}$ with $P$ being the transmit power per antenna. The phase control vector for the RIS in the $b$-th time slot is denoted by $\boldsymbol{\omega}_{b} \in \mathbb{C}^{N_\mathrm{R} \times 1}$, of which each entry is unit-modulus, i.e., $\left|[\boldsymbol{\omega}_{b}]_{n_\mathrm{R}}\right| = 1$, $\forall n_\mathrm{R} \in \{ 1,\cdots,N_\mathrm{R} \}$, since the passive reflecting elements in the RIS can only adjust their phases in practice \cite{huang2019reconfigurable,he2021channel}. In addition, the synchronization and pilot detection are assumed to be properly solved, e.g., via the methods developed in \cite{meyr1998digital} or \cite{pun2006maximum}. In the $b$-th training slot, the received signal at the UE is given by
\begin{equation} \label{Z_b}
	\mathbf{Z}_{b} = \mathbf{H}^{b}_\mathrm{BU} \mathbf{S}_{b} + \mathbf{N}_{b},
\end{equation}
where $\mathbf{H}^b_\mathrm{BU} \in \mathbb{C}^{N_\mathrm{U} \times N_\mathrm{B}}$ denotes the cascaded end-to-end channel between the BS and UE in the $b$-th time slot with $\boldsymbol{\omega}_{b}$ taken into account and $\mathbf{N}_{b} \in \mathbb{C}^{N_\mathrm{U} \times M}$ denotes the additive white Gaussian noise (AWGN) whose entries are independent and identically distributed (i.i.d.) following $\mathcal{CN}(0, \sigma^{2})$ with $\sigma^{2}$ being the noise variance. By exploiting the orthogonality of the training sequences, we right-multiply \eqref{Z_b} by $\mathbf{S}_{b}^{H} / P$, which yields the processed signal as
\begin{equation} \label{Y_b}
	\mathbf{Y}_{b} =\mathbf{H}^b_\mathrm{BU} + \mathbf{N}_{b}', 
\end{equation}
where $\mathbf{N}_{b}' = \frac{1}{P} \mathbf{N}_{b}\mathbf{S}_{b}^{H} \in \mathbb{C}^{N_{\mathrm{U}}\times N_{\mathrm{B}}}$ denotes a noise matrix with i.i.d. entries following $\mathcal{CN}(0, \frac{\sigma^{2}}{P})$. Then, the CE problem for the considered RIS-aided MIMO system in this paper is to estimate the cascaded channel $\mathbf{H}^b_\mathrm{BU}$ based on the processed signal $\mathbf{Y}_{b}$.

\subsection{Channel Model}

In this paper, we consider a geometry-based statistic channel model \cite{saleh1987statistical}, which has been commonly adopted in the RIS literature \cite{ardah2021trice,he2021channel,he2021leveraging,chung2021location} and includes various communication environments such as line-of-sight (LOS) and non-line-of-sight (NLOS) scenarios. We first define the steering vector of an $N$-element ULA with half-wavelength spacing as
\begin{equation}
	\mathbf{a}_{N} ( \theta ) = \big[ 1,\mathrm{e}^{i \pi \cos \theta},\cdots,\mathrm{e}^{i \pi ( N-1 ) \cos \theta} \big]^{T} \in \mathbb{C}^{N \times 1},
\end{equation}
where $\theta$ denotes the steering angle. Now, the BS-to-RIS channel $\mathbf{H}_\mathrm{BR} \in \mathbb{C}^{N_\mathrm{R} \times N_\mathrm{B}}$ is modeled as
\begin{equation} \label{HBR}
	\begin{aligned}
		\mathbf{H}_\mathrm{BR} & = \sum_{l_\mathrm{BR}=1}^{L_\mathrm{BR}} [\boldsymbol{\rho}_\mathrm{BR}]_{l_\mathrm{BR}} \mathbf{a}_{N_\mathrm{R}} ( [\boldsymbol{\phi}_\mathrm{R}]_{l_\mathrm{BR}} ) \mathbf{a}^{H}_{N_\mathrm{B}} ( [\boldsymbol{\theta}_\mathrm{B}]_{l_\mathrm{BR}} ) \\ & = \mathbf{A}_{N_\mathrm{R}} ( \boldsymbol{\phi}_\mathrm{R} ) \mathrm{diag} ( \boldsymbol{\rho}_\mathrm{BR} ) \mathbf{A}^{H}_{N_\mathrm{B}} ( \boldsymbol{\theta}_\mathrm{B} ),
	\end{aligned}
\end{equation}
where $L_\mathrm{BR}$ denotes the number of paths between the BS and RIS, variables $[\boldsymbol{\theta}_\mathrm{B}]_{l_\mathrm{BR}}, [\boldsymbol{\phi}_\mathrm{R}]_{l_\mathrm{BR}} \in [0,\pi]$\footnote{In this paper, we assume that the left-right ambiguity issue of the ULA can be resolved through array signal processing approaches \cite{van2004optimum,NairOCEANS} and thus the azimuths are assumed within a half angular space.}, and $ [\boldsymbol{\rho}_\mathrm{BR}]_{l_\mathrm{BR}} \in \mathbb{C}$ denote the angle-of-departure (AoD), AoA, and propagation path gain of the $l_\mathrm{BR}$-th path of the BS-to-RIS channel, respectively, $\mathbf{A}_{N_\mathrm{B}} ( \boldsymbol{\theta}_\mathrm{B} ) = \big[ \mathbf{a}_{N_\mathrm{B}} ( [\boldsymbol{\theta}_\mathrm{B}]_{1} ),\cdots,\mathbf{a}_{N_\mathrm{B}} ( [\boldsymbol{\theta}_\mathrm{B}]_{L_\mathrm{BR}} ) \big] \in \mathbb{C}^{N_\mathrm{B} \times L_\mathrm{BR}}$, and $\mathbf{A}_{N_\mathrm{R}} ( \boldsymbol{\phi}_\mathrm{R} ) = \big[ \mathbf{a}_{N_\mathrm{R}} ( [\boldsymbol{\phi}_\mathrm{R}]_{1} ),\cdots,\mathbf{a}_{N_\mathrm{R}} ( [\boldsymbol{\phi}_\mathrm{R}]_{L_\mathrm{BR}} ) \big] \in \mathbb{C}^{N_\mathrm{R} \times L_\mathrm{BR}}$. We assume that $\mathbf{H}_\mathrm{BR}$ is sparse in the angular domain in the sense that the number of paths $L_\mathrm{BR}$ is small. Similarly, the RIS-to-UE channel $\mathbf{H}_\mathrm{RU} \in \mathbb{C}^{N_\mathrm{U} \times N_\mathrm{R}}$ is modeled as
\begin{equation} \label{HRU}
	\begin{aligned}
		\mathbf{H}_\mathrm{RU} & = \sum_{l_\mathrm{RU}=1}^{L_\mathrm{RU}} [\boldsymbol{\rho}_\mathrm{RU}]_{l_\mathrm{RU}} \mathbf{a}_{N_\mathrm{U}} ( [\boldsymbol{\phi}_\mathrm{U}]_{l_\mathrm{RU}} ) \mathbf{a}^{H}_{N_\mathrm{R}} ( [\boldsymbol{\theta}_\mathrm{R}]_{l_\mathrm{RU}} ) \\ & = \mathbf{A}_{N_\mathrm{U}} ( \boldsymbol{\phi}_\mathrm{U} ) \mathrm{diag} ( \boldsymbol{\rho}_\mathrm{RU} ) \mathbf{A}^{H}_{N_\mathrm{R}} ( \boldsymbol{\theta}_\mathrm{R} ),
	\end{aligned}
\end{equation}
where $L_\mathrm{RU}$ denotes the number of paths between the RIS and UE, variables $[\boldsymbol{\theta}_\mathrm{R}]_{l_\mathrm{RU}}, [\boldsymbol{\phi}_\mathrm{U}]_{l_\mathrm{RU}} \in [0,\pi]$, and $[\boldsymbol{\rho}_\mathrm{RU}]_{l_\mathrm{RU}} \in \mathbb{C}$ denote the AoD, AoA, and propagation path gain of the $l_\mathrm{RU}$-th path of the RIS-to-UE channel, respectively, $\mathbf{A}_{N_\mathrm{R}} ( \boldsymbol{\theta}_\mathrm{R} ) = \big[ \mathbf{a}_{N_\mathrm{R}} ( [\boldsymbol{\theta}_\mathrm{R}]_{1} ),\cdots,\mathbf{a}_{N_\mathrm{R}} ( [\boldsymbol{\theta}_\mathrm{R}]_{L_\mathrm{RU}} ) \big] \in \mathbb{C}^{N_\mathrm{R} \times L_\mathrm{RU}}$, and $\mathbf{A}_{N_\mathrm{U}} ( \boldsymbol{\phi}_\mathrm{U} ) = \big[ \mathbf{a}_{N_\mathrm{U}} ( [\boldsymbol{\phi}_\mathrm{U}]_{1} ),\cdots,\mathbf{a}_{N_\mathrm{U}} ( [\boldsymbol{\phi}_\mathrm{U}]_{L_\mathrm{RU}} ) \big] \in \mathbb{C}^{N_\mathrm{U} \times L_\mathrm{RU}}$. We also assume that $\mathbf{H}_\mathrm{RU}$ is sparse in the angular domain with $L_\mathrm{RU}$ being small. Now, the cascaded channel $\mathbf{H}^b_\mathrm{BU} \in \mathbb{C}^{N_\mathrm{U} \times N_\mathrm{B}}$ is given by
\begin{equation} \label{HbBU}
	\begin{aligned}
		\mathbf{H}^b_\mathrm{BU} & = \mathbf{H}_\mathrm{RU} \mathrm{diag} ( \boldsymbol{\omega}_b ) \mathbf{H}_\mathrm{BR} \\ 
		& = \mathbf{A}_{N_\mathrm{U}} ( \boldsymbol{\phi}_\mathrm{U} ) \mathrm{diag} ( \boldsymbol{\rho}_\mathrm{RU} ) \mathbf{A}^{H}_{N_\mathrm{R}} ( \boldsymbol{\theta}_\mathrm{R} ) \mathrm{diag} ( \boldsymbol{\omega}_b ) \mathbf{A}_{N_\mathrm{R}} ( \boldsymbol{\phi}_\mathrm{R} ) \mathrm{diag} ( \boldsymbol{\rho}_\mathrm{BR} ) \mathbf{A}^{H}_{N_\mathrm{B}} ( \boldsymbol{\theta}_\mathrm{B} ).
	\end{aligned}
\end{equation}

It is observed from \eqref{HbBU} that the cascaded channel $\mathbf{H}^b_\mathrm{BU}$ exhibits four-dimensional angular sparsity, i.e., $L_\mathrm{BR}$ and $L_\mathrm{RU}$ are much smaller than the size of arrays at the BS, UE, and RIS, which can be exploited to facilitate the CE for RIS-aided MIMO systems. Nevertheless, there are still a large number of channel parameters to be estimated to reconstruct $\mathbf{H}^b_\mathrm{BU}$. Even worse, they are severely coupled with each other, which motivates us to propose a decoupled CE method that maintains a high CE accuracy with a reduced computational complexity\footnote{We focus on the CE of RIS-aided MIMO systems in this paper. After acquiring an estimate of the channel, the BS can use the methods developed in \cite{zhang2020capacity} or \cite{pan2020multicell} to design the optimal phase shifts for RIS.}.

\section{Partially Decoupled ANM-based Channel Estimation Framework} \label{SecPDANM}

In this section, we first reveal the inherent parameter ambiguity in the cascaded channel model and derive an effective channel model. Based on this model, we formulate the CE problem as a PDANM problem, which exploits the 3D angular sparsity of the effective channel in a partially decoupled manner. Since the problem at hand is intractable, we formulate an SDP problem that provides a lower bound for the original PDANM problem and is equivalent to PDANM under mild conditions. A detailed computational complexity analysis is further provided. In particular, the proposed PDANM-based CE framework serves as a building block for further improving CE accuracy and reducing training overhead in the following sections.

\subsection{Effective Channel Model}

To begin with, we reveal the parameter ambiguity in the cascaded channel model in \eqref{HbBU}, as stated in the following proposition.
\begin{proposition} \label{Proposition1}
	{\rm (Parameter Ambiguity)} From a given $\mathbf{H}^b_\mathrm{BU}$, the parameters $\boldsymbol{\phi}_\mathrm{R}$, $\boldsymbol{\theta}_\mathrm{R}$, $\boldsymbol{\rho}_\mathrm{BR}$, and $\boldsymbol{\rho}_\mathrm{RU}$ in \eqref{HbBU} cannot be uniquely identified, i.e., there are different sets of parameters satisfying
	\begin{equation}
		\begin{aligned}
			\mathbf{H}^b_\mathrm{BU} & = \mathbf{A}_{N_\mathrm{U}} ( \boldsymbol{\phi}_\mathrm{U} ) \mathrm{diag} ( \boldsymbol{\rho}_\mathrm{RU} ) \mathbf{A}^{H}_{N_\mathrm{R}} ( \boldsymbol{\theta}_\mathrm{R} )\mathrm{diag} ( \boldsymbol{\omega}_b ) \mathbf{A}_{N_\mathrm{R}} ( \boldsymbol{\phi}_\mathrm{R} ) \mathrm{diag} ( \boldsymbol{\rho}_\mathrm{BR} ) \mathbf{A}^{H}_{N_\mathrm{B}} ( \boldsymbol{\theta}_\mathrm{B} ) \\ & = \mathbf{A}_{N_\mathrm{U}} ( \boldsymbol{\phi}_\mathrm{U} ) \mathrm{diag} ( \boldsymbol{\rho}'_\mathrm{RU} ) \mathbf{A}^{H}_{N_\mathrm{R}} ( \boldsymbol{\theta}'_\mathrm{R} ) \mathrm{diag} ( \boldsymbol{\omega}_b ) \mathbf{A}_{N_\mathrm{R}} ( \boldsymbol{\phi}'_\mathrm{R} ) \mathrm{diag} ( \boldsymbol{\rho}'_\mathrm{BR} ) \mathbf{A}^{H}_{N_\mathrm{B}} ( \boldsymbol{\theta}_\mathrm{B} ),
		\end{aligned}
	\end{equation}
	where $\boldsymbol{\rho}'_\mathrm{RU} \ne \boldsymbol{\rho}_\mathrm{RU}$, $\boldsymbol{\theta}'_\mathrm{R} \ne \boldsymbol{\theta}_\mathrm{R}$, $\boldsymbol{\phi}'_\mathrm{R} \ne \boldsymbol{\phi}_\mathrm{R}$, and $\boldsymbol{\rho}'_\mathrm{BR} \ne \boldsymbol{\rho}_\mathrm{BR}$.
\end{proposition}
\begin{proof}
	See Appendix \ref{Proposition1Proof}.
\end{proof}

Proposition \ref{Proposition1} reveals the parameter ambiguity in $\mathbf{H}^{b}_\mathrm{BU}$, i.e., neither the path gains nor the angular parameters in the cascaded channel model can be uniquely identified, which may potentially affect some parameter-based applications in RIS-aided MIMO systems, e.g., beamforming. In contrast, the next proposed effective channel model avoids the potential harm caused by parameter ambiguity. Inserting \eqref{HbBU} into \eqref{Y_b} and vectorizing $\mathbf{Y}_{b}$, we have
\begin{equation} \label{y_b}
	\begin{aligned}
		\mathbf{y}_{b} & = \mathrm{vec} (\mathbf{Y}_{b}) = \mathrm{vec} \big( \mathbf{H}_\mathrm{RU} \mathrm{diag} ( \boldsymbol{\omega}_{b} ) \mathbf{H}_\mathrm{BR} \big) + \mathrm{vec} (\mathbf{N}_{b}') \\ & \mathop {=}\limits^{\left(\mathrm{a}\right)} (\mathbf{H}_\mathrm{BR}^{T} \diamond \mathbf{H}_\mathrm{RU}) \boldsymbol{\omega}_{b} + \mathbf{n}_{b} \\ & \mathop {=}\limits^{\left(\mathrm{b}\right)} \hspace{-1mm} \big[ \overline{\mathbf{A}_{N_\mathrm{B}} ( \boldsymbol{\theta}_\mathrm{B} )} \hspace{-0.5mm} \otimes \hspace{-0.5mm} \mathbf{A}_{N_\mathrm{U}} ( \boldsymbol{\phi}_\mathrm{U} ) \big] \big[ \mathrm{diag} ( \boldsymbol{\rho}_\mathrm{BR} ) \hspace{-0.5mm} \otimes \hspace{-0.5mm} \mathrm{diag} ( \boldsymbol{\rho}_\mathrm{RU} ) \big] \big[ \mathbf{A}^{T}_{N_\mathrm{R}} ( \boldsymbol{\phi}_\mathrm{R} ) \hspace{-0.5mm} \diamond \hspace{-0.5mm} \mathbf{A}^{H}_{N_\mathrm{R}} ( \boldsymbol{\theta}_\mathrm{R} ) \big] \boldsymbol{\omega}_{b} + \mathbf{n}_{b},
	\end{aligned}
\end{equation}
where $\left(\mathrm{a}\right)$ and $\left(\mathrm{b}\right)$ in \eqref{y_b} are obtained by exploiting the properties of Kronecker and Khatri-Rao product \cite{rao1998matrix}, respectively, and $\mathbf{n}_{b} = \mathrm{vec} (\mathbf{N}_{b}') \in \mathbb{C}^{N_{\mathrm{U}} N_{\mathrm{B}}\times1}$. Furthermore, to resolve the parameter ambiguity illustrated in Proposition \ref{Proposition1}, we define the effective path gains as $\boldsymbol{\rho}_\mathrm{BU} = \boldsymbol{\rho}_\mathrm{BR} \otimes \boldsymbol{\rho}_\mathrm{RU} \in \mathbb{C}^{L_\mathrm{BR} L_\mathrm{RU}\times1}$ and the \emph{differential angles} $\boldsymbol{\psi}_\mathrm{R}$ at the RIS as follows:
\begin{equation}
	\begin{aligned}
		\boldsymbol{\psi}_\mathrm{R} = \big\{ & [\boldsymbol{\psi}_\mathrm{R}]_{l_\mathrm{BU}} \in [0,\pi] : \cos ([\boldsymbol{\psi}_\mathrm{R}]_{l_\mathrm{BU}}) = \big[ \cos ([\boldsymbol{\theta}_\mathrm{R}]_{l_\mathrm{RU}}) - \cos ([\boldsymbol{\phi}_\mathrm{R}]_{l_\mathrm{BR}}) \big] \hspace{-0.1in} \mod 1, \\ & l_\mathrm{BR} = 1,\cdots,L_\mathrm{BR}, \ l_\mathrm{RU} = 1,\cdots,L_\mathrm{RU}, \ l_\mathrm{BU} = (l_{\mathrm{BR}}-1)L_{\mathrm{RU}}+l_\mathrm{RU} \big\}.
	\end{aligned}
\end{equation}
Then, we define the effective channel $\mathbf{H}$ between the BS and UE as
\begin{equation} \label{H}
	\begin{aligned}
		\mathbf{H} & = \mathbf{H}_\mathrm{BR}^{T} \diamond \mathbf{H}_\mathrm{RU} \\ & = \big[ \mathbf{A}_{N_\mathrm{B}} ( -\boldsymbol{\theta}_\mathrm{B} ) \otimes \mathbf{A}_{N_\mathrm{U}} ( \boldsymbol{\phi}_\mathrm{U} ) \big] \mathrm{diag} ( \boldsymbol{\rho}_\mathrm{BU} ) \mathbf{A}^{H}_{N_\mathrm{R}} ( \boldsymbol{\psi}_\mathrm{R} ) \\ & = \sum_{l_\mathrm{BR}=1}^{L_\mathrm{BR}} \sum_{l_\mathrm{RU}=1}^{L_\mathrm{RU}} [\boldsymbol{\rho}_\mathrm{BU}]_{l_\mathrm{BU}} \big[ \mathbf{a}_{N_\mathrm{B}} ( [-\boldsymbol{\theta}_\mathrm{B}]_{l_\mathrm{BR}} ) \otimes \mathbf{a}_{N_\mathrm{U}} ( [\boldsymbol{\phi}_\mathrm{U}]_{l_\mathrm{RU}} ) \big] \mathbf{a}^{H}_{N_\mathrm{R}} ( [\boldsymbol{\psi}_\mathrm{R}]_{l_\mathrm{BU}} ),
	\end{aligned}
\end{equation}
where $l_\mathrm{BU} = (l_{\mathrm{BR}}-1)L_{\mathrm{RU}}+l_\mathrm{RU}$ and $\mathbf{A}_{N_\mathrm{R}} ( \boldsymbol{\psi}_\mathrm{R} ) = \big[ \mathbf{a}_{N_\mathrm{R}} ( [\boldsymbol{\psi}_\mathrm{R}]_{1} ),\cdots,\mathbf{a}_{N_\mathrm{R}} ( [\boldsymbol{\psi}_\mathrm{R}]_{L_{\mathrm{BR}}L_\mathrm{RU}} ) \big] \in \mathbb{C}^{N_\mathrm{R} \times L_\mathrm{BR} L_\mathrm{RU}}$. With \eqref{H}, the vectorized signal $\mathbf{y}_{b}$ in \eqref{y_b} can be rewritten as
\begin{equation} \label{y_b2}
	\begin{aligned}
		\mathbf{y}_{b} & = \mathbf{H} \boldsymbol{\omega}_{b} + \mathbf{n}_{b}.
	\end{aligned}
\end{equation}
By stacking the vectorized signals $\mathbf{y}_{b}$ in all the $B$ time slots into a received signal matrix, we have
\begin{equation} \label{Y}
	\mathbf{Y} = \mathbf{H} \boldsymbol{\Omega} + \mathbf{N} \in \mathbb{C}^{N_\mathrm{B}N_\mathrm{U} \times B},
\end{equation}
where $\boldsymbol{\Omega} = \big[ \boldsymbol{\omega}_{1},\cdots,\boldsymbol{\omega}_{B} \big] \in \mathbb{C}^{N_\mathrm{R} \times B}$ denotes the RIS phase control matrix and $\mathbf{N} \in \mathbb{C}^{N_\mathrm{B}N_\mathrm{U} \times B}$ is the effective noise with i.i.d. entries following $\mathcal{CN}(0, \frac{\sigma^{2}}{P} )$. Now, the CE problem of the considered RIS-aided MIMO system is to estimate the effective channel $\mathbf{H}$ from the received signal $\mathbf{Y}$ with the known phase control matrix $\boldsymbol{\Omega}$.

It is observed from \eqref{Y} that the effective channel $\mathbf{H}$ is decoupled from the phases of RIS, which exhibits a simpler structure that is beneficial for CE. In particular, it is seen from \eqref{H} that $\mathbf{H}$ exhibits a 3D angular sparsity, i.e., the number of effective paths $L_\mathrm{BR}L_\mathrm{RU}$ is much smaller than $\min \{ N_\mathrm{B}N_\mathrm{U},N_\mathrm{R} \}$, where the $\big( (l_{\mathrm{BR}}-1)L_{\mathrm{RU}}+l_\mathrm{RU} \big)$-th effective path corresponding to an effective path gain $[\boldsymbol{\rho}_\mathrm{BR}]_{l_\mathrm{BR}} [\boldsymbol{\rho}_\mathrm{RU}]_{l_\mathrm{RU}}$, an AoD $[\boldsymbol{\theta}_\mathrm{B}]_{l_\mathrm{BR}}$ from the BS, an AoA $[\boldsymbol{\phi}_\mathrm{U}]_{l_\mathrm{RU}}$ to the UE, and a differential angle $[\boldsymbol{\psi}_\mathrm{R}]_{(l_{\mathrm{BR}}-1)L_{\mathrm{RU}}+l_\mathrm{RU}}$ at the RIS. On the other hand, the effective channel model does not suffer from the parameter ambiguity issue, thus is beneficial for applications requiring estimated channel parameters.

We note that a model similar to \eqref{H} has been considered in \cite{ardah2021trice} to separate the channel parameters at the BS and UE from those at the RIS and estimate them separately. The effective channel \eqref{H} is adopted in \cite{chung2021location} for CE. Differently from \cite{ardah2021trice,chung2021location}, in this paper, we have shown theoretically the necessity of exploiting the effective channel model for CE in the sense that the parameters cannot be identified and thus effectively estimated from the original cascaded channel model.

\subsection{Partially Decoupled ANM} \label{SubSecPDANM}

Since the effective channel $\mathbf{H}$ exhibits a 3D angular sparsity, the considered CE problem can be formulated as a sparse optimization problem. In particular, the \emph{atomic norm} \cite{chandrasekaran2012convex} is a class of sparsity metrics capable of exploiting the channel structure. In the following, we introduce some basics of atomic set and ANM before proposing the PDANM framework for the CE problem.

If $\mathbf{H}$ can be represented as a linear combination of some elements in a set with some coefficients, we call this set an \emph{atomic set} and this representation an \emph{atomic decomposition}. The \emph{atomic norm} \cite{chandrasekaran2012convex} of $\mathbf{H}$ with respect to an atomic set is defined as the minimum sum of the absolute values of each coefficient among all its atomic decompositions regarding this atomic set. For example, the 2D atomic set \cite{chung2021location} is defined as $\mathcal{A}_{\mathrm{2D}} = \Big\{ \big[ \mathbf{a}_{N_\mathrm{B}} ( \theta ) \otimes \mathbf{a}_{N_\mathrm{U}} ( \phi ) \big] \mathbf{b}^{H} : \theta, \phi \in [0,\pi], \| \mathbf{b} \|_{2} = 1 \Big\}$ and the corresponding 2D atomic norm of $\mathbf{H}$ is defined as
\begin{equation} \label{AN2D}
	\| \mathbf{H} \|_{\mathcal{A}_{\mathrm{2D}}} = \inf \Big\{ \sum_{l} | \rho_{l} | : \mathbf{H} = \sum_{l} \rho_{l} \big[ \mathbf{a}_{N_\mathrm{B}} ( \theta_{l} ) \otimes \mathbf{a}_{N_\mathrm{U}} ( \phi_{l} ) \big] \mathbf{b}_{l}^{H}, \theta_{l}, \phi_{l} \in [0,\pi], \| \mathbf{b}_{l} \|_{2} = 1 \Big\},
\end{equation}
which only exploits the 2D angular structure of the effective channel $\mathbf{H}$ but neglects its third-dimensional angular structure. To fully exploit the angular structure of the effective channel $\mathbf{H}$, \cite{chung2021location} further proposed to vectorize $\mathbf{H}$ and define the 3D atomic set as $\mathcal{A}_{\mathrm{3D}} = \big\{ \mathbf{a}_{N_\mathrm{R}} ( \psi ) \otimes \mathbf{a}_{N_\mathrm{B}} ( \theta ) \otimes \mathbf{a}_{N_\mathrm{U}} ( \phi ): \psi, \theta, \phi \in [0,\pi] \big\}$. Then, the corresponding 3D atomic norm of the vectorized effective channel $\mathrm{vec} (\mathbf{H})$ is defined as
\begin{equation} \label{AN3D}
	\| \mathrm{vec} (\mathbf{H}) \|_{\mathcal{A}_{\mathrm{3D}}} = \hspace{-0.5mm} \inf \Big\{ \sum_{l} | \rho_{l} | : \mathrm{vec} (\mathbf{H}) =\sum_{l} \rho_{l} \mathbf{a}_{N_\mathrm{R}} ( \psi_{l} ) \otimes \mathbf{a}_{N_\mathrm{B}} ( \theta_{l} ) \otimes \mathbf{a}_{N_\mathrm{U}} ( \phi_{l} ) , \psi_{l}, \theta_{l}, \phi_{l} \in [0,\pi] \Big\},
\end{equation}
which characterizes the structure of $\mathbf{H}$ more precisely than \eqref{AN2D} due to exploring its 3D angular structure. Although \eqref{AN3D} is intractable, $\| \mathrm{vec} (\mathbf{H}) \|_{\mathcal{A}_{\mathrm{3D}}}$ can be calculated via solving the following SDP problem \cite{yang2016vandermonde,chi2014compressive}
\begin{equation} \label{AN3Dopt}
	\begin{aligned}
			\min_{t,\mathbf{T}} & \frac{1}{2} t + \frac{1}{2N_\mathrm{R}N_\mathrm{B}N_\mathrm{U}} \mathrm{tr} (\mathcal{T}_{[N_\mathrm{R},N_\mathrm{B},N_\mathrm{U}]}(\mathbf{T})) \\ \mathrm{s.t.} &
			\begin{bmatrix}
					t & \mathrm{vec} (\mathbf{H})^{H}\\
					\mathrm{vec} (\mathbf{H}) & \mathcal{T}_{[N_\mathrm{R},N_\mathrm{B},N_\mathrm{U}]}(\mathbf{T})
				\end{bmatrix}
			\succeq \mathbf{0},
		\end{aligned}
\end{equation}
where $\mathcal{T}_{[N_\mathrm{R},N_\mathrm{B},N_\mathrm{U}]}(\mathbf{T}) \in \mathbb{C}^{N_\mathrm{R}N_\mathrm{B}N_\mathrm{U} \times N_\mathrm{R}N_\mathrm{B}N_\mathrm{U}}$ is a $3$-level Toeplitz matrix, as defined in Section \ref{SecInt}.

Due to the vectorization of $\mathbf{H}$, a large-scale SDP problem needs to be solved for calculating $\| \mathrm{vec} (\mathbf{H}) \|_{\mathcal{A}_{\mathrm{3D}}}$, which leads to a prohibitively high computational complexity. For channels with only 2D angular structures, such as \eqref{HBR} and \eqref{HRU}, a decoupled atomic norm and a decoupled ANM method was proposed in \cite{zhang2019efficient,xi2017super}, which reduces the computational complexity by decoupling the angular parameters in two different dimensions and solving a smaller-sized SDP problem. This motivates us to partially decouple the 3D angular parameters into two groups and reformulate the considered CE problem in \eqref{Y} as a PDANM problem. For this purpose, we define the partially decoupled atomic set as $\mathcal{A} = \big\{ \big[ \mathbf{a}_{N_\mathrm{B}} ( \theta ) \otimes \mathbf{a}_{N_\mathrm{U}} ( \phi ) \big] \mathbf{a}^{H}_{N_\mathrm{R}} ( \psi ) : \theta, \phi, \psi \in [0,\pi] \big\}$ and the PDAN as
\begin{equation} \label{PDAN}
	\| \mathbf{H} \|_{\mathcal{A}} = \inf \Big\{ \sum_{l} | \rho_{l} | : \mathbf{H} = \sum_{l} \rho_{l} \big[ \mathbf{a}_{N_\mathrm{B}} ( \theta_{l} ) \otimes \mathbf{a}_{N_\mathrm{U}} ( \phi_{l} ) \big] \mathbf{a}^{H}_{N_\mathrm{R}} ( \psi_{l} ), \theta_{l}, \phi_{l}, \psi_{l} \in [0,\pi]\Big\}.
\end{equation}
It is seen from the definitions of $\mathcal{A}$ and $\| \mathbf{H} \|_{\mathcal{A}}$ that they efficiently exploit the 3D angular structure of the effective channel. In addition, it is observed that $\| \mathrm{vec} (\mathbf{H}) \|_{\mathcal{A}_{\mathrm{3D}}} = \| \mathbf{H} \|_{\mathcal{A}}$ by comparing \eqref{AN3D} and \eqref{PDAN}. In particular, the calculation of $\| \mathbf{H} \|_{\mathcal{A}}$ can be realized by solving a smaller-size SDP problem than \eqref{AN3Dopt} in certain scenarios, as will be discussed below.

Similar to \eqref{AN3D}, \eqref{PDAN} is an intractable problem, thus we formulate an SDP problem to calculate $\| \mathbf{H} \|_{\mathcal{A}}$ in the following theorem:

\begin{figure*}[!t]
	\vspace{-3mm}
	\centering
	\subfloat[]
	{\label{fig3a}\includegraphics[width=2in]{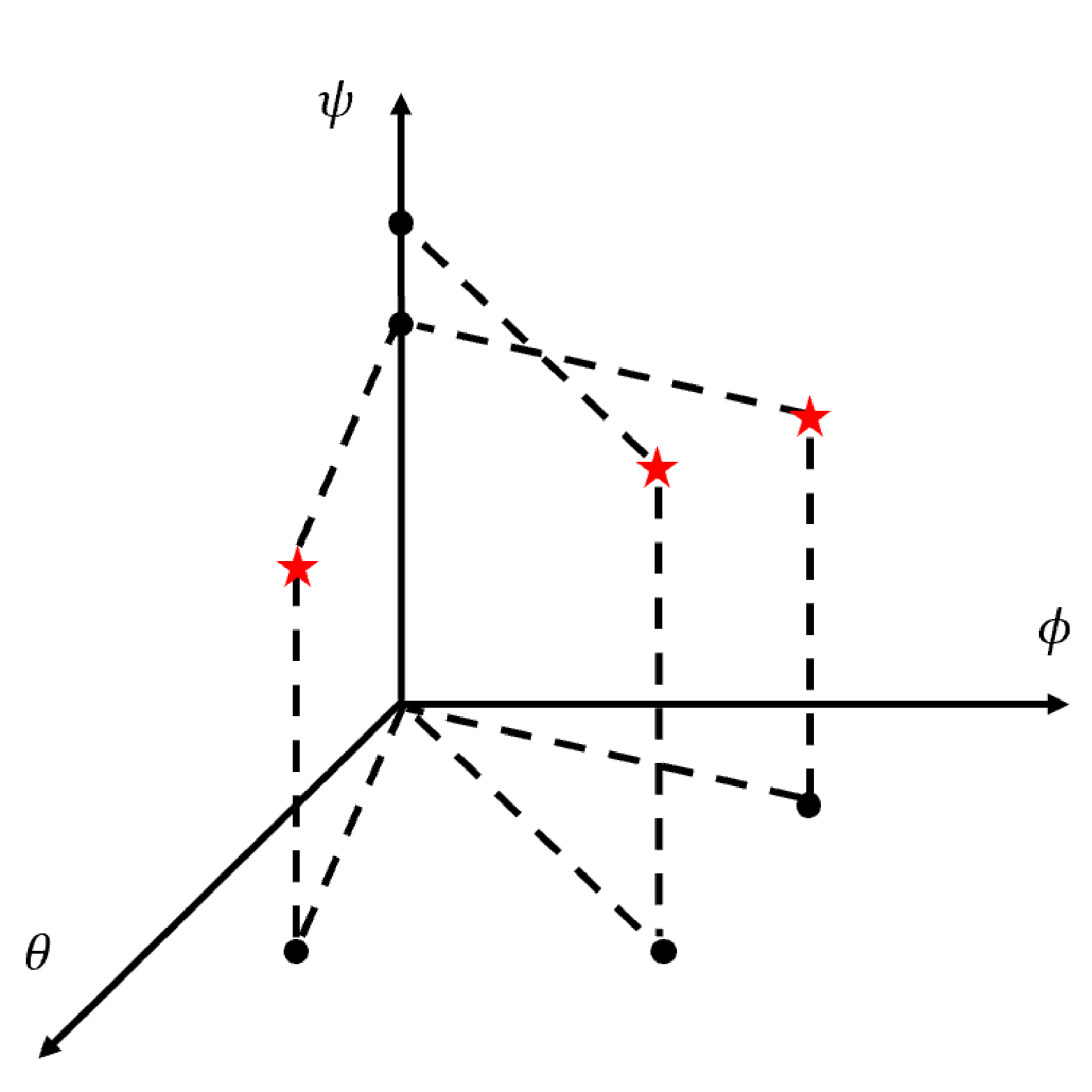}}
	\subfloat[]
	{\label{fig3b}\includegraphics[width=2in]{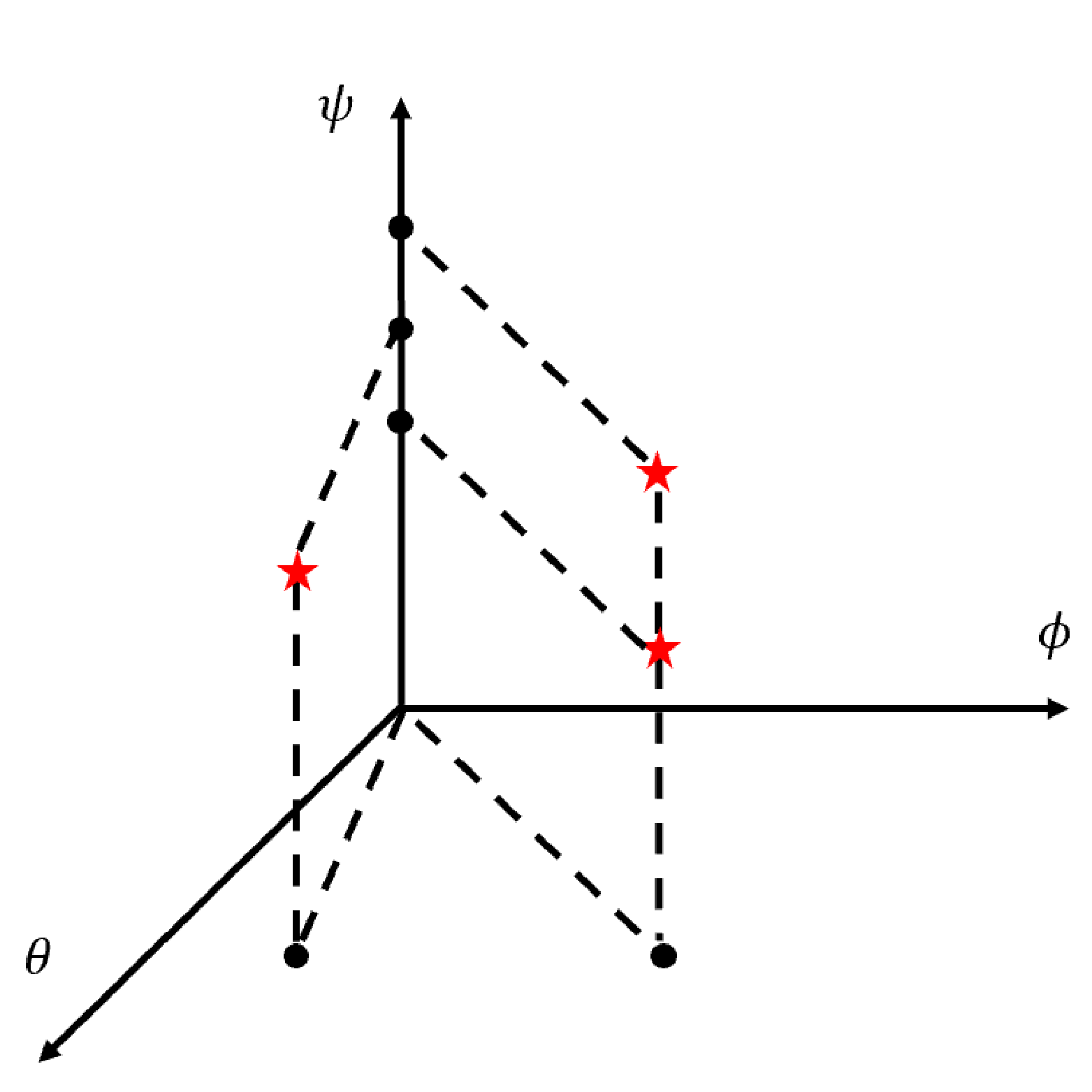}}
	\subfloat[]
	{\label{fig3c}\includegraphics[width=2in]{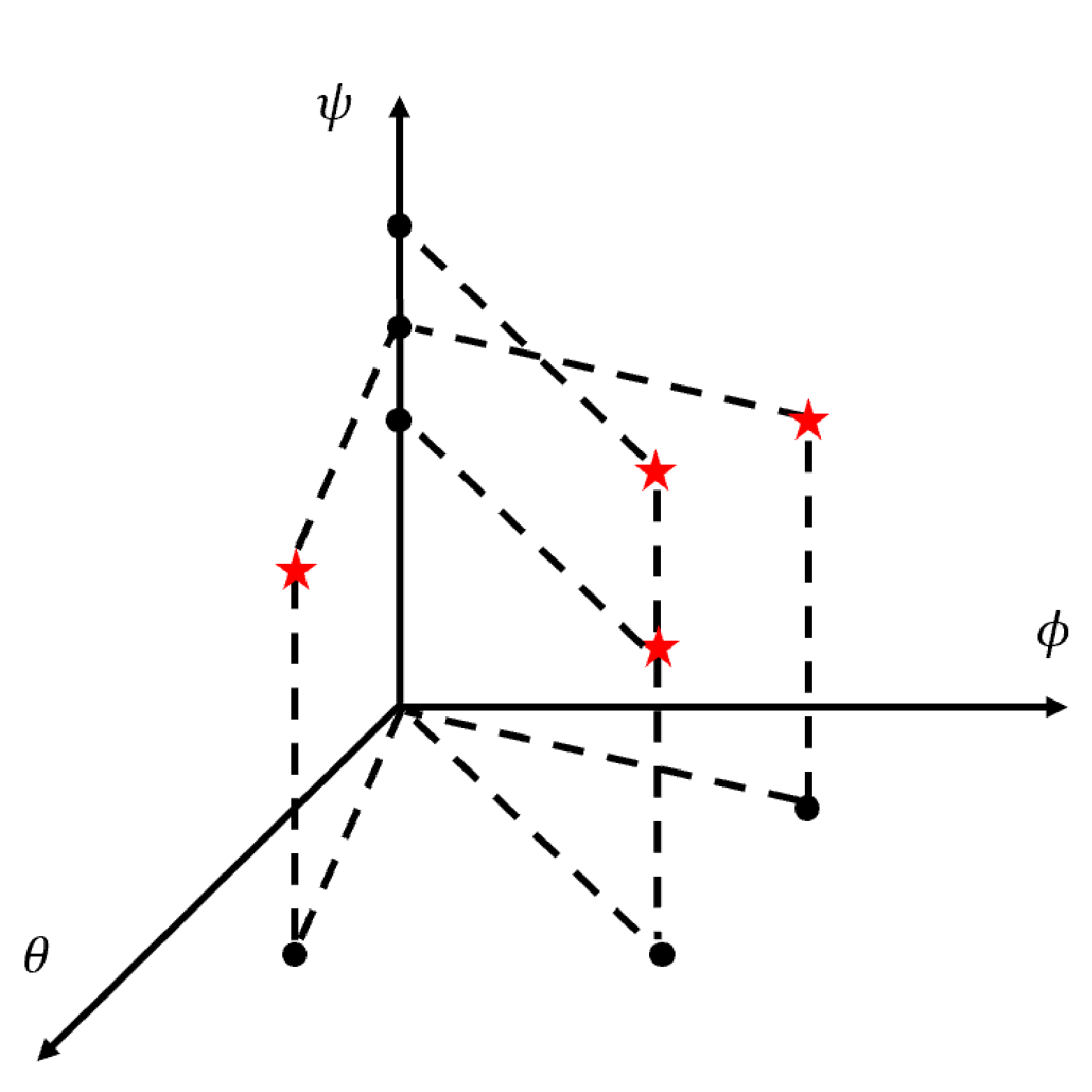}}
	\caption{An example of groups of partially decoupled atoms in several cases, where the red pentagrams correspond to the partially decoupled atoms and the black dots denote their projections on the $\{ \theta,\phi \}$ plane and $\psi$ axis. The three partially decoupled atoms in Fig. \ref{fig3a} are separable on the $\{ \theta,\phi \}$ plane with $L_\mathrm{R}^{*} = 2$ and $L_\mathrm{BU}^{*} = 3$ and the three partially decoupled atoms in Fig. \ref{fig3b} are separable on the $\psi$ axis with $L_\mathrm{R}^{*} = 3$ and $L_\mathrm{BU}^{*} = 2$. The four partially decoupled atoms in Fig. \ref{fig3c} overlap on both the $\{ \theta,\phi \}$ plane and the $\psi$ axis with $L_\mathrm{R}^{*} = L_\mathrm{BU}^{*} = 3$.}
	\label{fig3}
	\vspace{-5mm}
\end{figure*}

\begin{theorem} \label{Theorem1}
	For a given effective channel $\mathbf{H}$, we have $\mathrm{SDP}(\mathbf{H}) \leq \| \mathbf{H} \|_{\mathcal{A}}$, where $\mathrm{SDP}(\mathbf{H})$ is the optimal value of the following SDP problem:
	\begin{equation} \label{PDANopt}
		\begin{aligned}
			\min_{\mathbf{t},\mathbf{T}} & \frac{1}{2N_\mathrm{R}} \mathrm{tr} \big(\mathcal{T}_{N_\mathrm{R}}(\mathbf{t})\big) + \frac{1}{2N_\mathrm{B}N_\mathrm{U}} \mathrm{tr} \big(\mathcal{T}_{[N_\mathrm{B},N_\mathrm{U}]}(\mathbf{T})\big) \\ \mathrm{s.t.} &
			\begin{bmatrix}
				\mathcal{T}_{N_\mathrm{R}}(\mathbf{t}) & \mathbf{H}^{H}\\
				\mathbf{H} & \mathcal{T}_{[N_\mathrm{B},N_\mathrm{U}]}(\mathbf{T})
			\end{bmatrix}
			\succeq \mathbf{0}.
		\end{aligned}
	\end{equation}
	Furthermore, denote $\{ \mathbf{t}^{*},\mathbf{T}^{*} \}$ as the optimizer of \eqref{PDANopt}, then we have $\mathrm{SDP}(\mathbf{H}) = \| \mathbf{H} \|_{\mathcal{A}}$ if the following conditions are satisfied:
	\begin{enumerate}
		\item{$\mathrm{rank} \big(\mathcal{T}_{N_\mathrm{R}}(\mathbf{t}^{*})\big) < N_\mathrm{R}$. It implies that $\mathcal{T}_{N_\mathrm{R}}(\mathbf{t}^{*})$ admits a Vandermonde decomposition \cite[Theorem 11.5]{yang2018sparse}, i.e.,
		\begin{equation} \label{VanDec1}
			\begin{aligned}
				\mathcal{T}_{N_\mathrm{R}}(\mathbf{t}^{*}) & = \sum_{l_\mathrm{R}=1}^{L_\mathrm{R}^{*}} [\mathbf{p}_\mathrm{R}^{*}]_{l_\mathrm{R}} \mathbf{a}_{N_\mathrm{R}} ( [\boldsymbol{\psi}_\mathrm{R}^{*}]_{l_\mathrm{R}} ) \mathbf{a}^{H}_{N_\mathrm{R}} ( [\boldsymbol{\psi}_\mathrm{R}^{*}]_{l_\mathrm{R}} ) \\ & = \mathbf{A}_{N_\mathrm{R}} ( \boldsymbol{\psi}_\mathrm{R}^{*} ) \mathrm{diag} (\mathbf{p}_\mathrm{R}^{*}) \mathbf{A}^{H}_{N_\mathrm{R}} ( \boldsymbol{\psi}_\mathrm{R}^{*} ),
			\end{aligned}
		\end{equation}
		where $L_\mathrm{R}^{*} = \mathrm{rank} \big(\mathcal{T}_{N_\mathrm{R}}(\mathbf{t}^{*})\big)$ is the number of estimated differential angles, $[\boldsymbol{\psi}_\mathrm{R}^{*}]_{l_\mathrm{R}}$ is the $l_\mathrm{R}$-th estimated differential angle, $[\mathbf{p}_\mathrm{R}^{*}]_{l_\mathrm{R}}$ is the amplitude of the $l_\mathrm{R}$-th estimated path gain, and $\mathbf{A}_{N_\mathrm{R}} ( \boldsymbol{\psi}_\mathrm{R}^{*} ) = \big[ \mathbf{a}_{N_\mathrm{R}} ( [\boldsymbol{\psi}_\mathrm{R}^{*}]_{1} ),\cdots,\mathbf{a}_{N_\mathrm{R}} ( [\boldsymbol{\psi}_\mathrm{R}^{*}]_{L_\mathrm{R}^{*}} ) \big]$ is full column rank.} \label{con1}
		\item{$\mathcal{T}_{[N_\mathrm{B},N_\mathrm{U}]}(\mathbf{T}^{*})$ admits a 2-level Vandermonde decomposition \cite{yang2016vandermonde}, i.e.,
		\begin{equation} \label{VanDec2}
			\begin{aligned}
				\mathcal{T}_{[N_\mathrm{B},N_\mathrm{U}]}(\mathbf{T}^{*}) & = \sum_{l_\mathrm{BU}=1}^{L_\mathrm{BU}^{*}} [\mathbf{p}_\mathrm{BU}^{*}]_{l_\mathrm{BU}} \big[ \mathbf{a}_{N_\mathrm{B}} ( [\boldsymbol{\theta}_\mathrm{B}^{*}]_{l_\mathrm{BU}} ) \otimes \mathbf{a}_{N_\mathrm{U}} ( [\boldsymbol{\phi}_\mathrm{U}^{*}]_{l_\mathrm{BU}} ) \big] \big[ \mathbf{a}_{N_\mathrm{B}} ( [\boldsymbol{\theta}_\mathrm{B}^{*}]_{l_\mathrm{BU}} ) \otimes \mathbf{a}_{N_\mathrm{U}} ( [\boldsymbol{\phi}_\mathrm{U}^{*}]_{l_\mathrm{BU}} ) \big]^{H} \\ & = \big[ \mathbf{A}_{N_\mathrm{B}} ( \boldsymbol{\theta}_\mathrm{B}^{*} ) \diamond \mathbf{A}_{N_\mathrm{U}} ( \boldsymbol{\phi}_\mathrm{U}^{*} ) \big] \mathrm{diag} (\mathbf{p}_\mathrm{BU}^{*}) \big[ \mathbf{A}_{N_\mathrm{B}} ( \boldsymbol{\theta}_\mathrm{B}^{*} ) \diamond \mathbf{A}_{N_\mathrm{U}} ( \boldsymbol{\phi}_\mathrm{U}^{*} ) \big]^{H},
			\end{aligned}
		\end{equation}
		where $L_\mathrm{BU}^{*} = \mathrm{rank} \big(\mathcal{T}_{[N_\mathrm{B},N_\mathrm{U}]}(\mathbf{T}^{*})\big) $ is the number of estimated AoD-AoA pairs, $[\boldsymbol{\theta}_\mathrm{B}^{*}]_{l_\mathrm{BU}}$ is the $l_\mathrm{BU}$-th estimated AoD, $[\boldsymbol{\phi}_\mathrm{U}^{*}]_{l_\mathrm{BU}}$ is the $l_\mathrm{BU}$-th estimated AoA, $[\mathbf{p}_\mathrm{BU}^{*}]_{l_\mathrm{BU}}$ is the amplitude of the $l_\mathrm{BU}$-th estimated path gain, $\mathbf{A}_{N_\mathrm{B}} ( \boldsymbol{\theta}_\mathrm{B}^{*} ) = \big[ \mathbf{a}_{N_\mathrm{B}} ( [\boldsymbol{\theta}_\mathrm{B}^{*}]_{1} ),\cdots,\mathbf{a}_{N_\mathrm{B}} ( [\boldsymbol{\theta}_\mathrm{B}^{*}]_{L_\mathrm{BU}^{*}} ) \big]$, $\mathbf{A}_{N_\mathrm{U}} ( \boldsymbol{\phi}_\mathrm{U}^{*} ) = \big[ \mathbf{a}_{N_\mathrm{U}} ( [\boldsymbol{\phi}_\mathrm{U}^{*}]_{1} ),\cdots,\mathbf{a}_{N_\mathrm{U}} ( [\boldsymbol{\phi}_\mathrm{U}^{*}]_{L_\mathrm{BU}^{*}} ) \big]$, and $\big[ \mathbf{A}_{N_\mathrm{B}} ( \boldsymbol{\theta}_\mathrm{B}^{*} ) \diamond \mathbf{A}_{N_\mathrm{U}} ( \boldsymbol{\phi}_\mathrm{U}^{*} ) \big]$ is full column rank.} \label{con2}
		\item{$\mathbf{C} = \big[ \mathbf{A}_{N_\mathrm{B}} ( \boldsymbol{\theta}_\mathrm{B}^{*} ) \diamond \mathbf{A}_{N_\mathrm{U}} ( \boldsymbol{\phi}_\mathrm{U}^{*} ) \big]^{\dagger} \mathbf{H} {\mathbf{A}^{H}_{N_\mathrm{R}} ( \boldsymbol{\psi}_\mathrm{R}^{*} )}^{\dagger}$ has only one nonzero element per row or per column.} \label{con3}
	\end{enumerate}
\end{theorem}

\begin{proof}
	See Appendix \ref{Theorem1Proof}.
\end{proof}

Theorem \ref{Theorem1} illustrates that $\mathrm{SDP}(\mathbf{H})$ is a lower bound of $\| \mathbf{H} \|_{\mathcal{A}}$ and proves the equivalence between \eqref{PDAN} and \eqref{PDANopt} under some mild sufficient conditions. In particular, the first two conditions in Theorem \ref{Theorem1} guarantee that a partially decoupled atomic decomposition of $\mathbf{H}$ can be obtained from the optimizer of \eqref{PDANopt}, while the third condition implies that the corresponding partially decoupled atoms have either non-overlapping AoD-AoA pairs or non-overlapping differential angles. An example of groups of partially decoupled atoms is given in Fig. \ref{fig3}, where the partially decoupled atoms in Fig. \ref{fig3a} and Fig. \ref{fig3b} are separable on the $\{ \theta,\phi \}$ plane or the $\psi$ axis and thus satisfy the conditions in Theorem \ref{Theorem1}, while those in Fig. \ref{fig3c} overlap on both the $\{ \theta,\phi \}$ plane and the $\psi$ axis and thus do not satisfy the conditions in Theorem \ref{Theorem1}. In addition, we expect that $L_\mathrm{R}^{*} = L_\mathrm{BU}^{*} = L_\mathrm{BR}L_\mathrm{RU}$ holds in the general case, which means that both the numbers of estimated differential angles and estimated AoD-AoA pairs are equal to that of actual effective paths.

Furthermore, a more intuitive sufficient condition for the equivalence between \eqref{PDAN} and \eqref{PDANopt} is that the effective channel paths are well separable in terms of differential angles, as stated in the following theorem:

\begin{theorem} \label{Theorem2}
	Define the minimum sinusoidal interval of a set of angles $\boldsymbol{\psi} = [ \psi_{1},\cdots,\psi_{K} ]$ as
	\begin{equation} \label{MinSinInt}
		\Delta_{\boldsymbol{\psi}} = \inf_{\psi_{i},\psi_{j}:i \ne j} \{ |\cos \psi_{i} - \cos \psi_{j}|,1 - |\cos \psi_{i} - \cos \psi_{j}| \}.
	\end{equation}
	If the actual differential angles $\boldsymbol{\psi}_\mathrm{R}$ satisfy that $\Delta_{\boldsymbol{\psi}_\mathrm{R}} > \frac{4}{N_\mathrm{R}}$, then we have $\mathrm{SDP}(\mathbf{H}) = \| \mathbf{H} \|_{\mathcal{A}}$.
\end{theorem}

\begin{proof}
	See Appendix \ref{Theorem2Proof}.
\end{proof}

Upon the condition that the actual differential angles $\boldsymbol{\psi}_\mathrm{R}$ are separable to a certain extent, Theorem \ref{Theorem2} establishes the equivalence between \eqref{PDAN} and \eqref{PDANopt}. This condition is expected to hold in practice, as RIS usually consists of a large number of elements. Note that the condition stated in Theorem \ref{Theorem2} pertains to the actual differential angles, while the conditions in Theorem \ref{Theorem1} concern the estimated differential angles and AoD-AoA pairs. In conclusion, \eqref{PDANopt} provides a lower bound for $\| \mathbf{H} \|_{\mathcal{A}}$ and is equal to $\| \mathbf{H} \|_{\mathcal{A}}$ under two classes of mild conditions in Theorem \ref{Theorem1} and Theorem \ref{Theorem2}, respectively.

Next, we formulate the considered CE problem as a PDANM problem, i.e., to treat $\mathbf{H}$ as an optimization variable and find an estimated channel matrix $\hat{\mathbf{H}}$ with the minimum PDAN. With Theorem \ref{Theorem1} and Theorem \ref{Theorem2}, a solution to the PDANM problem can be obtained by solving the following problem:
\begin{equation} \label{PDANMopt}
	\begin{aligned}
		\min_{\mathbf{t},\mathbf{T},\mathbf{H}} & \frac{1}{2N_\mathrm{R}} \mathrm{tr} \big(\mathcal{T}_{N_\mathrm{R}}(\mathbf{t})\big) + \frac{1}{2N_\mathrm{B}N_\mathrm{U}} \mathrm{tr} \big(\mathcal{T}_{[N_\mathrm{B},N_\mathrm{U}]}(\mathbf{T})\big) \\ \mathrm{s.t.} &
		\begin{bmatrix}
			\mathcal{T}_{N_\mathrm{R}}(\mathbf{t}) & \mathbf{H}^{H}\\
			\mathbf{H} & \mathcal{T}_{[N_\mathrm{B},N_\mathrm{U}]}(\mathbf{T})
		\end{bmatrix}
		\succeq \mathbf{0}, \| \mathbf{Y} - \mathbf{H} \boldsymbol{\Omega} \|_{F}^{2} \leq \eta,
	\end{aligned}
\end{equation}
where $\eta$ is a given constant proportional to $\sigma^2$ \cite{yang2016exact}. Since \eqref{PDANMopt} is an SDP problem, it can be solved by the interior-point method via an SDP solver, such as SDPT3 \cite{toh1999sdpt3}. Our proposed framework is named PDANM since it essentially decouples the 3D angular structure of the effective channel into two low-dimensional angular structures, i.e., AoD-AoA pairs and differential angles.

Similarly, ANM-2D and ANM-3D find an $\hat{\mathbf{H}}$ with the minimum 2D atomic norm and 3D atomic norm for the CE by solving an SDP problem similar to \eqref{PDANMopt}, respectively. However, according to the analysis of the 2D atomic norm and 3D atomic norm above, ANM-2D does not fully exploit the 3D angular structure of $\mathbf{H}$, while ANM-3D needs to solve a large-scale SDP problem. Instead, the proposed PDANM framework not only fully exploits the 3D angular structure, but also reduces the computational complexity (to be illustrated in the next subsection), since the size of the PSD matrix in \eqref{PDANopt} is much smaller than that in \eqref{AN3Dopt}. In addition, for the considered CE problem, it is difficult to define a fully decoupled atomic norm that is calculable.

\subsection{Complexity Analysis} \label{SubSecCom}

\begin{table*}[!t]
	\caption{Computational Complexity of Various Methods \label{ComANM}}
	\renewcommand{\arraystretch}{1.2}
	\centering
	\begin{tabular}{c|c|c|c}
		\hline
		& Variable Size & Matrix Size & Computational Complexity \\
		\hline
		PDANM & $\mathcal{O}(N_\mathrm{B}N_\mathrm{U}N_\mathrm{R})$ & $N_\mathrm{B}N_\mathrm{U} + N_\mathrm{R}$ & $\mathcal{O}\big((N_\mathrm{B}N_\mathrm{U}N_\mathrm{R})^{2}(N_\mathrm{B}N_\mathrm{U} + N_\mathrm{R})^{2.5}\big)$ \\
		\hline
		ANM-2D \cite{chung2021location} & $\mathcal{O}(N_\mathrm{B}N_\mathrm{U}N_\mathrm{R} + N_\mathrm{R}^{2})$ & $N_\mathrm{B}N_\mathrm{U} + N_\mathrm{R}$ & $\mathcal{O}\big((N_\mathrm{B}N_\mathrm{U}N_\mathrm{R} + N_\mathrm{R}^{2})^{2}(N_\mathrm{B}N_\mathrm{U} + N_\mathrm{R})^{2.5}\big)$ \\
		\hline
		ANM-3D \cite{chung2021location} & $\mathcal{O}(N_\mathrm{B}N_\mathrm{U}N_\mathrm{R})$ & $N_\mathrm{B}N_\mathrm{U}N_\mathrm{R}$ & $\mathcal{O}\big((N_\mathrm{B}N_\mathrm{U}N_\mathrm{R})^{4.5}\big)$ \\
		\hline
		KRF \cite{de2021channel} & \diagbox{}{} & \diagbox{}{} & $\mathcal{O}(N_\mathrm{B}N_\mathrm{U}N_\mathrm{R})$ \\
		\hline
		TRICE \cite{ardah2021trice} & \diagbox{}{} & \diagbox{}{} & $\mathcal{O}\big(L_\mathrm{BR}L_\mathrm{RU}(N_\mathrm{B}N_\mathrm{U}L_\mathrm{BR}^{2}L_\mathrm{RU}^{2}+N_\mathrm{B}N_\mathrm{U}\bar{N}_\mathrm{B}\bar{N}_\mathrm{U}+N_\mathrm{R}\bar{N}_\mathrm{R})\big)$ \\
		\hline
	\end{tabular}
	\vspace{-3mm}
\end{table*}

In this subsection, we analyze the computational complexity of the proposed PDANM method and compare it with the state-of-the-art CE methods. According to \cite{vandenberghe1996semidefinite}, the computational complexity of solving an SDP problem by the interior-point method is $\mathcal{O}\big(N_\mathrm{1}^{2} N_\mathrm{2}^{2.5}\big)$, where $N_\mathrm{1}$ is the variable size and $N_\mathrm{2} \times N_\mathrm{2}$ is the size of the PSD matrix. Since the dimension of the PSD matrix is $N_\mathrm{B}N_\mathrm{U} + N_\mathrm{R}$ and the variable size is on the order of $N_\mathrm{B}N_\mathrm{U}N_\mathrm{R}$ in \eqref{PDANMopt}, the computational complexity of PDANM is $\mathcal{O}\big((N_\mathrm{B}N_\mathrm{U}N_\mathrm{R})^{2}(N_\mathrm{B}N_\mathrm{U} + N_\mathrm{R})^{2.5}\big)$.

By a similar analysis, the computational complexity of the ANM-2D approach and the ANM-3D in \cite{chung2021location} is derived and summarized in Table \ref{ComANM}. For intuitive comparison, we calculate the ratio of complexity of ANM-2D and ANM-3D to that of PDANM, which are $\mathcal{O}\Big(\big(1+\frac{N_\mathrm{R}}{N_\mathrm{B}N_\mathrm{U}}\big)^2\Big)$ and $\mathcal{O}\Big(\big(\frac{N_\mathrm{B}N_\mathrm{U}N_\mathrm{R}}{N_\mathrm{B}N_\mathrm{U} + N_\mathrm{R}}\big)^{2.5}\Big)$, respectively. It is observed that the complexity of the proposed PDANM is lower than that of ANM-2D and is significantly lower than that of ANM-3D. In practice, the running time of solving an SDP problem depends heavily on the matrix size rather than the variable size. As a result, ANM-3D suffers a long running time due to the large size of the PSD matrix, as will be verified in Section \ref{SecNum}.

In addition, we include the computational complexity of KRF \cite{de2021channel} and TRICE \cite{ardah2021trice} in Table \ref{ComANM} for comparison, where $\bar{N}_\mathrm{B}$, $\bar{N}_\mathrm{U}$, and $\bar{N}_\mathrm{R}$ denote the grid resolution of AoDs at BS, AoAs at UE, and differential angles at RIS, respectively. It is seen that the complexity of KRF and TRICE is lower than the ANM-based methods, and the complexity of TRICE grows as the increase of path number while other methods do not suffer from this issue.

\section{Reweighted PDANM Algorithm} \label{SecRPDANM}

To improve CE accuracy, in this section we propose an iterative algorithm named reweighted PDANM (RPDANM), which is inspired by the reweighted ANM algorithm proposed in our previous work \cite{yang2015enhancing,chu2023new}\footnote{Note that the proposed RPDANM algorithm is fundamentally different from \cite{yang2015enhancing,chu2023new}. Since the considered CE problem can be regarded as a 3D angular parameter estimation problem, the methods for one-dimensional angular parameter estimation developed in \cite{yang2015enhancing,chu2023new} are not applicable to our problem.}.

It is observed from \eqref{PDANMopt} that PDANM promotes the sparsity of the effective channel by solving a trace minimization problem. A formulation that promotes sparsity more efficiently is the rank minimization problem \cite{yang2016exact} obtained by substituting the trace in \eqref{PDANMopt} for the rank:
\begin{equation} \label{PDrankminopt}
	\begin{aligned}
		\min_{\mathbf{t},\mathbf{T},\mathbf{H}} & \frac{1}{2N_\mathrm{R}} \mathrm{rank} \big(\mathcal{T}_{N_\mathrm{R}}(\mathbf{t})\big) \hspace{-1mm} + \hspace{-1mm} \frac{1}{2N_\mathrm{B}N_\mathrm{U}} \mathrm{rank} \big(\mathcal{T}_{[N_\mathrm{B},N_\mathrm{U}]}(\mathbf{T})\big) \\ \mathrm{s.t.} &
		\begin{bmatrix}
			\mathcal{T}_{N_\mathrm{R}}(\mathbf{t}) & \mathbf{H}^{H}\\
			\mathbf{H} & \mathcal{T}_{[N_\mathrm{B},N_\mathrm{U}]}(\mathbf{T})
		\end{bmatrix}
		\succeq \mathbf{0}, \| \mathbf{Y} - \mathbf{H} \boldsymbol{\Omega} \|_{F}^{2} \leq \eta,
	\end{aligned}
\end{equation}
which is non-convex and NP-hard. In fact, \eqref{PDANMopt} is the convex relaxation of \eqref{PDrankminopt}. To further promote sparsity, we propose to approximately solve \eqref{PDrankminopt} via solving a set of SDP problems (see \cite{yang2015enhancing,chu2023new} for details). The proposed algorithm is named RPDANM since the problem solved in each iteration is a weighted PDANM (WPDANM) problem with varying weighting matrices and functions, as will be detailed below.

To begin with, we define the weighted PDAN. For positive definite weighting matrices $\mathbf{W}_{\mathrm{BU}}  \in \mathbb{C}^{N_\mathrm{B}N_\mathrm{U} \times N_\mathrm{B}N_\mathrm{U}}$ and $\mathbf{W}_{\mathrm{R}} \in \mathbb{C}^{N_\mathrm{R} \times N_\mathrm{R}}$, the weighting functions are defined as
\begin{equation}
	\begin{aligned}
		w_{\mathrm{BU}}(\theta,\phi) & = \Big\{ \big[ \mathbf{a}_{N_\mathrm{B}} ( \theta ) \otimes \mathbf{a}_{N_\mathrm{U}} ( \phi ) \big]^{H} \mathbf{W}_{\mathrm{BU}} \big[ \mathbf{a}_{N_\mathrm{B}} ( \theta ) \otimes \mathbf{a}_{N_\mathrm{U}} ( \phi ) \big] \Big\}^{-1} \quad \text{and} \\ w_{\mathrm{R}}(\psi) & = \big[ \mathbf{a}^{H}_{N_\mathrm{R}} ( \psi ) \mathbf{W}_{\mathrm{R}} \mathbf{a}_{N_\mathrm{R}} ( \psi ) \big]^{-1},
	\end{aligned}
\end{equation}
respectively. Furthermore, we define the weighted partially decoupled atomic set as
\begin{equation} \label{WPDAS}
	\mathcal{A}^{w_{\mathrm{BU}},w_{\mathrm{R}}} = \Big\{ w_{\mathrm{BU}}(\theta,\phi)^{\frac{1}{2}} w_{\mathrm{R}}(\psi)^{\frac{1}{2}} \big[ \mathbf{a}_{N_\mathrm{B}} ( \theta ) \otimes \mathbf{a}_{N_\mathrm{U}} ( \phi ) \big] \mathbf{a}^{H}_{N_\mathrm{R}} ( \psi ) : \theta, \phi, \psi \in [0,\pi] \Big\}
\end{equation}
and the weighted PDAN of $\mathbf{H}$ as
\begin{equation} \label{WPDAN}
	\begin{aligned}
		\| \mathbf{H} \|_{\mathcal{A}^{w_{\mathrm{BU}},w_{\mathrm{R}}}} = & \inf \Big\{ \sum_{l} | \rho_{l} |:\mathbf{H} =\sum_{l} \rho_{l} w_{\mathrm{BU}}(\theta_{l},\phi_{l})^{\frac{1}{2}} w_{\mathrm{R}}(\psi_{l})^{\frac{1}{2}} \\ & \cdot \big[ \mathbf{a}_{N_\mathrm{B}} ( \theta_{l} ) \otimes \mathbf{a}_{N_\mathrm{U}} ( \phi_{l} ) \big] \mathbf{a}^{H}_{N_\mathrm{R}} ( \psi_{l} ) , \theta_{l}, \phi_{l}, \psi_{l} \in [0,\pi] \Big\} \\ = & \inf \Big\{ \sum_{l} w_{\mathrm{BU}}(\theta_{l},\phi_{l})^{-\frac{1}{2}} w_{\mathrm{R}}(\psi_{l})^{-\frac{1}{2}} | \rho_{l} |: \\ & \mathbf{H} =\sum_{l} \rho_{l} \big[ \mathbf{a}_{N_\mathrm{B}} ( \theta_{l} ) \otimes \mathbf{a}_{N_\mathrm{U}} ( \phi_{l} ) \big] \mathbf{a}^{H}_{N_\mathrm{R}} ( \psi_{l} ) , \theta_{l}, \phi_{l}, \psi_{l} \in [0,\pi] \Big\},
	\end{aligned}
\end{equation}
which is a generalization of the PDAN defined in \eqref{PDAN} and degenerates to PDAN under specific weighting functions. Similar to \eqref{PDANopt}, we formulate an SDP problem to calculate $\| \mathbf{H} \|_{\mathcal{A}^{w_{\mathrm{BU}},w_{\mathrm{R}}}}$ in the following theorem:

\begin{theorem} \label{Theorem3}
	For a given effective channel $\mathbf{H}$, we have $\mathrm{SDP}_{\mathbf{W}_{\mathrm{R}},\mathbf{W}_{\mathrm{BU}}}(\mathbf{H}) \leq \| \mathbf{H} \|_{\mathcal{A}^{w_{\mathrm{R}},w_{\mathrm{BU}}}}$ with $\mathrm{SDP}_{\mathbf{W}_{\mathrm{R}},\mathbf{W}_{\mathrm{BU}}}(\mathbf{H})$ being the optimal value of the following SDP problem:
	\begin{equation} \label{WPDANopt}
		\begin{aligned}
			\min_{\mathbf{t},\mathbf{T}} & \frac{1}{2} \mathrm{tr} \big(\mathbf{W}_{\mathrm{R}}\mathcal{T}_{N_\mathrm{R}}(\mathbf{t})\big) + \frac{1}{2} \mathrm{tr} \big(\mathbf{W}_{\mathrm{BU}}\mathcal{T}_{[N_\mathrm{B},N_\mathrm{U}]}(\mathbf{T})\big) \\ \mathrm{s.t.} &
			\begin{bmatrix}
				\mathcal{T}_{N_\mathrm{R}}(\mathbf{t}) & \mathbf{H}^{H}\\
				\mathbf{H} & \mathcal{T}_{[N_\mathrm{B},N_\mathrm{U}]}(\mathbf{T})
			\end{bmatrix}
			\succeq \mathbf{0},
		\end{aligned}
	\end{equation}
	Furthermore, denote $\{ \mathbf{t}^{*},\mathbf{T}^{*} \}$ as the optimizer of \eqref{WPDANopt}, then we have $\mathrm{SDP}_{\mathbf{W}_{\mathrm{R}},\mathbf{W}_{\mathrm{BU}}}(\mathbf{H}) = \| \mathbf{H} \|_{\mathcal{A}^{w_{\mathrm{R}},w_{\mathrm{BU}}}}$ if the conditions \ref{con1}), \ref{con2}), and \ref{con3}) in Theorem \ref{Theorem1} hold.
\end{theorem}

\begin{proof}
	See Appendix \ref{Theorem3Proof}.
\end{proof}

By taking $\mathbf{W}_{\mathrm{BU}} = \frac{1}{N_\mathrm{B}N_\mathrm{U}} \mathbf{I}_{N_\mathrm{B}N_\mathrm{U}}$ and $\mathbf{W}_{\mathrm{R}} = \frac{1}{N_\mathrm{R}} \mathbf{I}_{N_\mathrm{R}}$, Theorem \ref{Theorem1} is a special case of Theorem \ref{Theorem3}. Then, similar to PDANM, we treat $\mathbf{H}$ as an optimization variable and find an $\hat{\mathbf{H}}$ with the minimum weighted PDAN for the CE. With Theorem \ref{Theorem3}, the WPDANM problem can be effectively solved by solving an SDP problem:
\begin{equation} \label{WPDANMopt}
	\begin{aligned}
		\min_{\mathbf{t},\mathbf{T},\mathbf{H}} & \frac{1}{2} \mathrm{tr} \big(\mathbf{W}_{\mathrm{R}}\mathcal{T}_{N_\mathrm{R}}(\mathbf{t})\big) + \frac{1}{2} \mathrm{tr} \big(\mathbf{W}_{\mathrm{BU}}\mathcal{T}_{[N_\mathrm{B},N_\mathrm{U}]}(\mathbf{T})\big) \\ \mathrm{s.t.} &
		\begin{bmatrix}
			\mathcal{T}_{N_\mathrm{R}}(\mathbf{t}) & \mathbf{H}^{H}\\
			\mathbf{H} & \mathcal{T}_{[N_\mathrm{B},N_\mathrm{U}]}(\mathbf{T})
		\end{bmatrix}
		\succeq \mathbf{0}, \| \mathbf{Y} - \mathbf{H} \boldsymbol{\Omega} \|_{F}^{2} \leq \eta,
	\end{aligned}
\end{equation}
which is a weighted version of \eqref{PDANMopt}. Due to the extra degrees of freedom introduced by weighting matrices, adopting appropriate weighting matrices will lead to a more accurate CE by \eqref{WPDANMopt} than \eqref{PDANMopt}.

Now, we propose an iterative RPDANM algorithm for CE of the considered RIS-aided MIMO system, which consists of solving a series of WPDANM problems with varying weighting matrices and functions. At the first iteration, we take initial weighting matrices $\mathbf{W}_{\mathrm{BU}} = \frac{1}{N_\mathrm{B}N_\mathrm{U}} \mathbf{I}_{N_\mathrm{B}N_\mathrm{U}}$ and $\mathbf{W}_{\mathrm{R}} = \frac{1}{N_\mathrm{R}} \mathbf{I}_{N_\mathrm{R}}$ and solve \eqref{WPDANMopt} to obtain the optimizer $\{ \mathbf{t}^{*},\mathbf{T}^{*},\hat{\mathbf{H}} \}$ with $\hat{\mathbf{H}}$ being an estimate of the effective channel. Then, we update the weighting matrices as
\begin{equation} \label{W}
	\begin{aligned}
		\mathbf{W}_{\mathrm{BU}} & = \big( \mathcal{T}_{[N_\mathrm{B},N_\mathrm{U}]}(\mathbf{T}^{*}) + \epsilon \mathbf{I}_{N_\mathrm{B}N_\mathrm{U}} \big)^{-1} \quad \text{and} \\ \mathbf{W}_{\mathrm{R}} & = \big( \mathcal{T}_{N_\mathrm{R}}(\mathbf{t}^{*}) + \epsilon \mathbf{I}_{N_\mathrm{R}} \big)^{-1}
	\end{aligned}
\end{equation}
and solve \eqref{WPDANMopt} again to update $\{ \mathbf{t}^{*},\mathbf{T}^{*},\hat{\mathbf{H}} \}$, where $\epsilon > 0$ is a regularization parameter that guarantees the invertibility of matrices. To gradually approach a solution to \eqref{PDrankminopt}, $\epsilon$ is halved after each iteration \cite{yang2015enhancing,chu2023new}. This step is repeated until the difference between two adjacent $\hat{\mathbf{H}}$ is smaller than a given threshold or the maximum number of iterations is reached. It is worth mentioning that each iteration of our proposed RPDANM algorithm aims to solve a WPDANM problem, while the first iteration of RPDANM is equivalent to PDANM. In fact, the iterative process can be regarded as continuously selecting better weighting matrices according to the last CE to gradually refine the CE. The proposed RPDANM algorithm does not require additional training overhead, yet the iterative process incurs a higher computational complexity. Specifically, with $T$ denoting the number of iterations that is generally small, RPDANM is $T$ times slower than PDANM since each iteration has the same computational complexity as that of PDANM.

\section{RPDANM with Adaptive Phase Control Approach} \label{SecRPDANMAPC}

Both PDANM and RPDANM follow the commonly adopted channel sounding procedure in the literature on RIS, as illustrated in Fig. \ref{fig2}. They both require a large number of training slots for realizing high CE accuracy, especially when the RIS consists of numerous elements. To facilitate low-overhead CE in the considered RIS-aided MIMO system, we consider the design of the phase control matrix $\boldsymbol{\Omega}$ for RIS during the training stage and propose a RPDANM with adaptive phase control (RPDANM-APC) approach for CE in this section.

\begin{algorithm}[!t]
	\caption{Proposed RPDANM-APC Approach}\label{Alg:RPDANM-APC}
	\begin{algorithmic}[1]
		\REQUIRE Noise variance $\sigma^2$, convergence threshold $\epsilon_{\mathbf{H}}$, initial number of training slots $B_{0}$, maximum number of training slots $B_{\max}$, phase control matrix $\boldsymbol{\Omega} \in \mathbb{C}^{N_\mathrm{R} \times B_{0}}$, regularization parameter $\epsilon = 1$, weighting matrices $\mathbf{W}_{\mathrm{BU}} = \frac{1}{N_\mathrm{B}N_\mathrm{U}} \mathbf{I}_{N_\mathrm{B}N_\mathrm{U}}$ and $\mathbf{W}_{\mathrm{R}} = \frac{1}{N_\mathrm{R}} \mathbf{I}_{N_\mathrm{R}}$. %%input
		\ENSURE Estimate of channel $\hat{\mathbf{H}}$, number of training slots $B$. %%output
		
		\STATE Set $B = B_{0}$ and $K = 0$.
		\STATE Send pilots in the first $B_{0}$ slots with the RIS phases set as each column of $\boldsymbol{\Omega} \in \mathbb{C}^{N_\mathrm{R} \times B_{0}}$, respectively, and collect the received data $\mathbf{Y} \in \mathbb{C}^{N_\mathrm{B}N_\mathrm{U} \times B_{0}}$.
		\STATE \label{stepSolveSDP} Solve the SDP problem in \eqref{WPDANMopt} to get its optimizer $\{ \mathbf{t}^{*},\mathbf{T}^{*},\hat{\mathbf{H}} \}$ and $L^{*}_{K} = \mathrm{rank}(\mathcal{T}_{N_\mathrm{R}}(\mathbf{t}^{*}))$.
		\STATE \label{stepVanDec} Obtain $\boldsymbol{\psi}_\mathrm{R}^{*}$ from the Vandermonde decomposition of $\mathcal{T}_{N_\mathrm{R}}(\mathbf{t}^{*})$ as in \eqref{VanDec1} via root-MUSIC \cite{rao1989performance}.
		\WHILE{$B + L^{*}_{K} \leq B_{\max}$}
		\IF{$\epsilon > \frac{\sigma^2}{10}$}
		\STATE Let $\epsilon = \frac{\epsilon}{2}$.
		\ENDIF
		\STATE Set $\hat{\mathbf{H}}_{\mathrm{last}} = \hat{\mathbf{H}}$, $B = B + L^{*}_{K}$.
		\STATE Send pilots in the following $L^{*}_{K}$ slots with the RIS phases set as each column of $\boldsymbol{\Omega}_{\mathrm{add}} = \mathbf{A}_{N_\mathrm{R}} ( \boldsymbol{\psi}_\mathrm{R}^{*} ) \in \mathbb{C}^{N_\mathrm{R} \times L^{*}_{K}}$, respectively, and collect the received data $\mathbf{Y}_{\mathrm{add}} \in \mathbb{C}^{N_\mathrm{B}N_\mathrm{U} \times L^{*}_{K}}$.
		\STATE Update $\mathbf{W}_{\mathrm{BU}}$ and $\mathbf{W}_{\mathrm{R}}$ as \eqref{W}, $\boldsymbol{\Omega} = \big[\boldsymbol{\Omega},\boldsymbol{\Omega}_{\mathrm{add}}\big]$, and $\mathbf{Y} = \big[\mathbf{Y},\mathbf{Y}_{\mathrm{add}}\big]$.
		\STATE Set $K = K + 1$ and repeat Step \ref{stepSolveSDP} and Step \ref{stepVanDec}.
		\IF{$\frac{\hat{\mathbf{H}}-\hat{\mathbf{H}}_{\mathrm{last}}\|_{F}^{2}}{\|\hat{\mathbf{H}}_{\mathrm{last}}\|_{F}^{2}} < \sigma^{2}\epsilon_{\mathbf{H}}$}
		\STATE \textbf{break}
		\ENDIF
		\ENDWHILE
	\end{algorithmic}
\end{algorithm}

In the following, we present the procedure of the proposed RPDANM-APC approach, including an initialization stage and an iteration stage, as summarized in Algorithm \ref{Alg:RPDANM-APC}. The initialization stage includes Steps 1 to 4, where the initial training sequences are transmitted in Step 2 and the effective channel and its parameters are estimated in Steps 3 and 4. The iteration stage includes Steps 5 to 16, where the steps of setting RIS phases, sending pilot and collecting received data, and channel parameter estimation are performed in each iteration. Note that these steps are repeated until convergence or the upper bound of training overhead is reached. Specifically, $\epsilon$ is halved until smaller than a given threshold, e.g., $2^{-10}$. In each slot of an iteration, the RIS phases are set as the steering vector corresponding to each estimated differential angle for coherently combining the signal from a BS-to-RIS path to a RIS-to-UE path. In Step 12, additional training sequences are employed for updating the CE with the specially set RIS phases\footnote{Note that the RIS requires feedback from the BS to optimize its phases, which affects the practicality of RPDANM-APC to some degree. However, since only the $L^{*}_{K}$ estimated angles instead of a whole phase control matrix need to be fed back at each iteration, the feedback overhead is generally less than expected and acceptable.}. Note that the number of training slots used in each iteration is equal to the number of estimated differential angles in the last iteration such that all estimated differential angles are explored. The convergence of the algorithm is checked in Steps 13 to 15, where $\epsilon_{\mathbf{H}}$ is a given constant accusing for the tolerance of converge accuracy.

By iteratively updating the RIS phases, sending pilots, and performing CE, the proposed RPDANM-APC approach takes full advantage of the ability of RIS to reshape the wireless propagation environment. Since the phases of RIS elements are adaptively adjusted to match the estimated effective paths, the passive beamforming gain of RIS can be fully exploited to improve CE accuracy. Moreover, as only part of the angular space is searched during iterations, RPDANM-APC achieves reduced training overhead compared to the PDANM and RPDANM methods. Note that the trade-off between training overhead and CE accuracy can be optimized by changing $B_{\max}$ and $\epsilon_{\mathbf{H}}$. The total number of training slots required by RPDANM-APC is $B = B_{0} + \sum_{k=0}^{K-1} L^{*}_{k}$ that is bounded by $ B_{\max}$, where $K$ is the total number of iterations and $L^{*}_{k}$ is the number of estimated differential angles at the $k$-th iteration. Since the computational complexity of the root-MUSIC algorithm in Step \ref{stepVanDec} is $\mathcal{O}\big(N_\mathrm{R}^{3}\big)$, the computational complexity of RPDANM-APC mainly comes from solving SDP problems in Step \ref{stepSolveSDP}, i.e., $\mathcal{O}\big(K(N_\mathrm{B}N_\mathrm{U}N_\mathrm{R})^{2}(N_\mathrm{B}N_\mathrm{U} + N_\mathrm{R})^{2.5}\big)$.

\section{Numerical Simulations} \label{SecNum}

\subsection{Simulation Settings}

In this section, we evaluate the performance of our proposed methods through numerical simulations. Unless otherwise stated, the default parameters are set as follows. We select $N_\mathrm{B} = 4$, $N_\mathrm{U} = 4$, $N_\mathrm{R} = 16$, and $L_{\mathrm{BR}} = L_{\mathrm{RU}} = 2$. All AoAs and AoDs are randomly and uniformly located in $[0,\pi]$ and all path gains are generated i.i.d. from $\mathcal{CN}(0,1)$. The signal-to-noise ratio (SNR) is defined as
\begin{equation}
	\text{SNR} = 10\log_{10} \frac{\|\mathbf{Y}-\mathbf{N}\|^{2}_{F}}{\|\mathbf{N}\|^{2}_{F}},
\end{equation}
which is set 30 dB as default. The convergence threshold is set $\epsilon_{\mathbf{H}} = 10^{-3}$ to strike a proper trade-off between the CE accuracy and the number of iterations. As in \cite{yang2016exact}, we set $\eta = \big(N_\mathrm{B}N_\mathrm{U}B+2\sqrt{N_\mathrm{B}N_\mathrm{U}B}\big) \sigma^{2}$ when solving a given SDP problem. For RPDANM, the maximum iteration number is 10. For RPDANM-APC, the initial and maximum number of training slots is $B_{0} = \frac{N_\mathrm{R}}{2}$ and $B_{\max} = N_\mathrm{R}$, respectively, while the number of training slots for all other methods is fixed as $N_\mathrm{R}$. Note that $B_{\max}$ is set to a large value here to observe the convergence performance of RPDANM-APC. In practice, $B_{\max}$ can be set to a small value to balance the required training overhead and CE accuracy. In addition, each entry of the (initial) phase control matrix $\boldsymbol{\Omega}$ is generated from the i.i.d. uniform distribution on a complex unit circle for all the considered methods, obeying the isotropy and incoherence properties in \cite{heckel2017generalized}.

\subsection{Channel Estimation Accuracy}

In this subsection, we compare the CE accuracy of the proposed methods, i.e., PDANM, RPDANM, and RPDANM-APC, to the state-of-the-art CE methods for RIS-aided MIMO systems. In particular, we adopt KRF \cite{de2021channel}, TRICE \cite{ardah2021trice}, ANM-2D \cite{chung2021location}, and ANM-3D \cite{chung2021location} for comparison, corresponding to the parallel factorization-based methods, CS methods, and ANM-based methods, respectively. The accuracy of CE is evaluated by the normalized mean square error (NMSE) that defined as
\begin{equation}
	\text{NMSE} = \mathbb{E} \left( \frac{\|\hat{\mathbf{H}}-\mathbf{H}\|_{F}^{2}}{\|\mathbf{H}\|_{F}^{2}} \right).
\end{equation}

\begin{figure}[!t]
	\centering
	\includegraphics[width=5in]{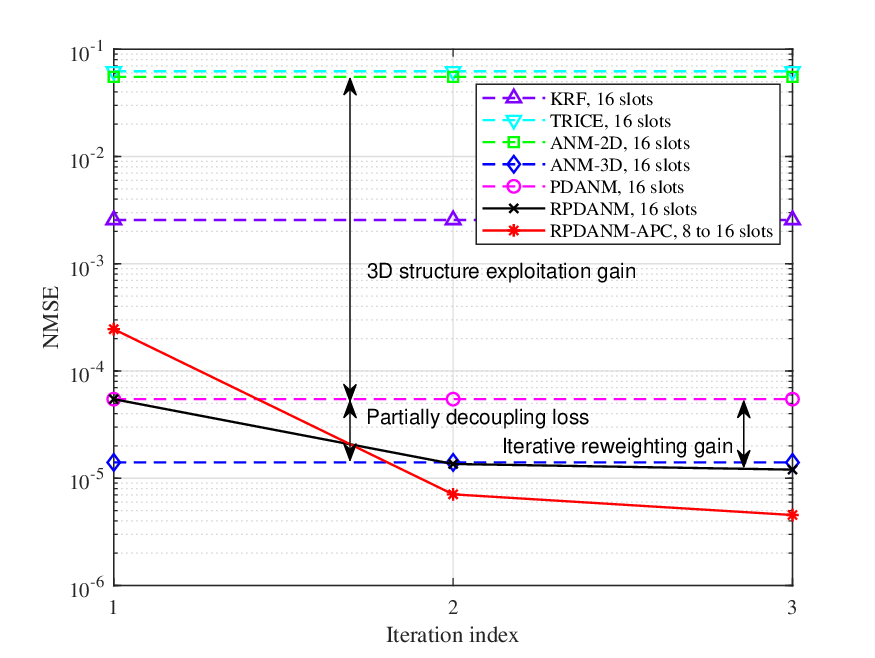}
	\vspace{-3mm}
	\caption{NMSE of different methods versus iteration index. Dashed lines denote one-step methods and solid curves denote iterative methods.}
	\label{fig4}
	\vspace{-4mm}
\end{figure}

An illustrative example is provided in Fig. \ref{fig4}, which depicts the NMSE of different methods versus the iteration index in a simulation. On the one hand, we observe that PDANM significantly outperforms ANM-2D due to the exploitation of the 3D angular structure of the effective channel. On the other hand, a slight accuracy loss of PDANM compared to ANM-3D is observed due to partially decoupling. Furthermore, the first iteration of RPDANM is equivalent to PDANM and thus they achieve the same accuracy. Besides, the CE accuracy of RPDANM is enhanced in subsequent iterations by continuously selecting better weighting matrices. In addition, the performance of RPDANM-APC is worse than that of RPDANM in the first iteration since fewer initial training slots are employed. However, as the number of training slots increases, the RPDANM-APC approach even outperforms RPDANM in terms of CE accuracy due to exploiting the RIS to reshape wireless propagation. In addition, the proposed methods have significantly higher accuracy over KRF and TRICE due to fully exploiting of the channel structure.

\begin{figure}[!t]
	\centering
	\includegraphics[width=5in]{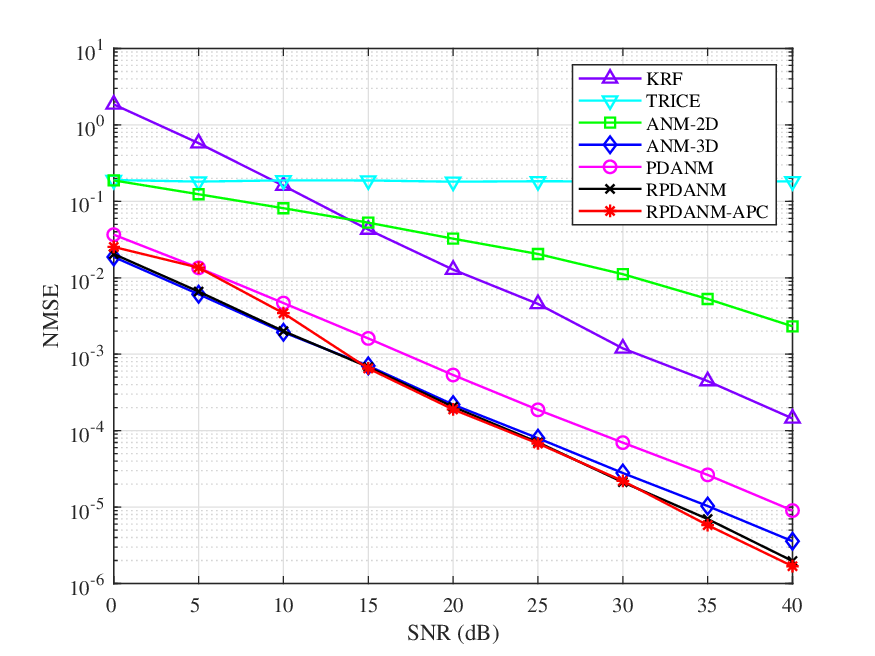}
	\vspace{-3mm}
	\caption{NMSE of different methods versus SNR.}
	\label{fig5}
	\vspace{-4mm}
\end{figure}

To compare the performance of the proposed methods comprehensively, Fig. \ref{fig5} illustrates the NMSE of different methods versus the SNR. We set 9 different SNR levels in the range of [0 dB, 40 dB] and conduct 50 random Monte Carlo experiments under each setting. It is observed that the CE accuracy of KRF steadily improves with the increase of SNR while a performance gap between it and our proposed methods exists. The performance of TRICE keeps almost unchanged with the variations of the SNR due to the grid mismatch problem. Among ANM-based methods, the PDANM is superior to ANM-2D due to the exploration of the 3D angular structure of the effective channel, but inferior to ANM-3D, owing to the accuracy loss caused by partially decoupling. By approximately solving the rank minimization problem via solving a series of WPDANM problems, RPDANM achieves higher CE accuracy than PDANM. Furthermore, RPDANM-APC outperforms other methods when $\text{SNR} > 15 \ \mathrm{dB}$. These simulation results demonstrate the excellent performance of our proposed methods, especially RPDANM and RPDANM-APC in terms of CE accuracy, consistent with our expectations and the example in Fig. \ref{fig4}.

\begin{figure}[!t]
	\centering
	\includegraphics[width=5in]{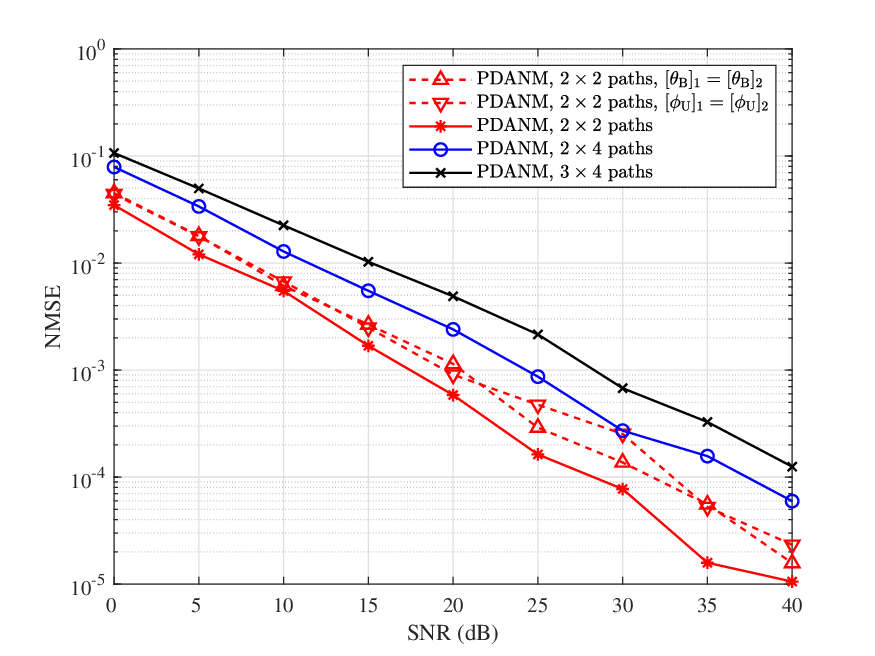}
	\vspace{-3mm}
	\caption{NMSE of PDANM versus SNR with different numbers of paths (solid curves) and in the presence of identical angles (dashed curves).}
	\label{fig6}
	\vspace{-4mm}
\end{figure}

To further illustrate the performance of the proposed methods, Fig. \ref{fig6} depicts the performance of PDANM with varying numbers of paths and in the presence of identical angles, where each point corresponds to the average performance of 50 Monte Carlo simulations. Since RPDANM and PDANM-APC exhibit a similar trend, the corresponding curves are not depicted for brevity. Comparing the three curves for 4, 8, and 12 effective paths, it is observed that the performance of PDANM decreases as the number of paths increases. Comparing the two dashed curves with 4 effective paths, the PDANM suffers a slight accuracy loss when there are paths with identical AoD at the BS or identical AoA at the UE. On the one hand, this demonstrates that the proposed PDANM still has superior parameter identification capability and achieve a high CE accuracy even in the presence of identical angular parameters, as stated in Theorem \ref{Theorem1}. On the other hand, since the angles in our simulations are randomly generated, the probability of occurrence of paths with close $\{ \theta,\phi \}$ is higher when AoDs at the BS or AoAs at the UE are fixed to be identical, leading to a decrease in the average NMSE performance.

\subsection{Running Time}

\begin{figure}[!t]
	\centering
	\includegraphics[width=5in]{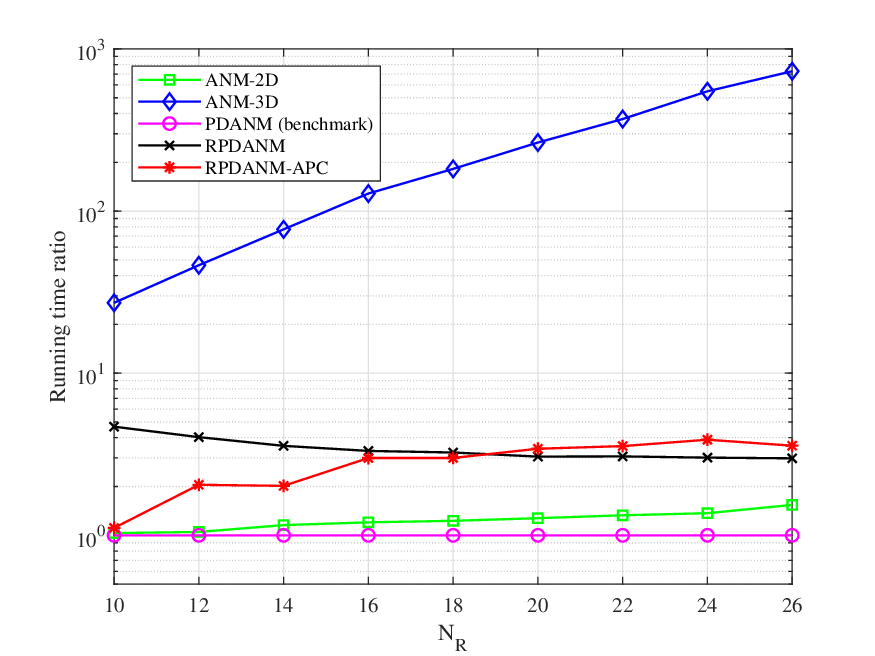}
	\vspace{-3mm}
	\caption{Running time ratios of different methods to the benchmark versus the size of RIS.}
	\label{fig7}
	\vspace{-4mm}
\end{figure}

In this subsection, we compare the running time of different methods versus the size of RIS. For a convenient comparison, Fig. \ref{fig7} depicts the ratio of running time of each method to our proposed PDANM approach versus the RIS size, where the running time ratio of the PDANM approach is normalized to one unit. The size of RIS is increased from 10 to 26 with 50 Monte Carlo simulations performed under each setting. The running time of ANM-3D increases rapidly as $N_\mathrm{R}$ increases, which is tens or hundreds of times longer than PDANM due to the requirement of solving a large-scale SDP problem. Besides, the running time of the proposed PDANM method is slightly shorter than ANM-2D due to fewer optimization variables involved. These simulation results are consistent with our computational complexity analysis in Section \ref{SubSecCom}. Furthermore, the running time of both RPDANM and RPDANM-APC are less than 5 times that of PDANM since only a few iterations are required for convergence. Therefore, our proposed three methods strike a good balance between CE accuracy and computational complexity compared with existing methods.

\subsection{Training Overhead}

\begin{figure}[!t]
	\centering
	\includegraphics[width=5in]{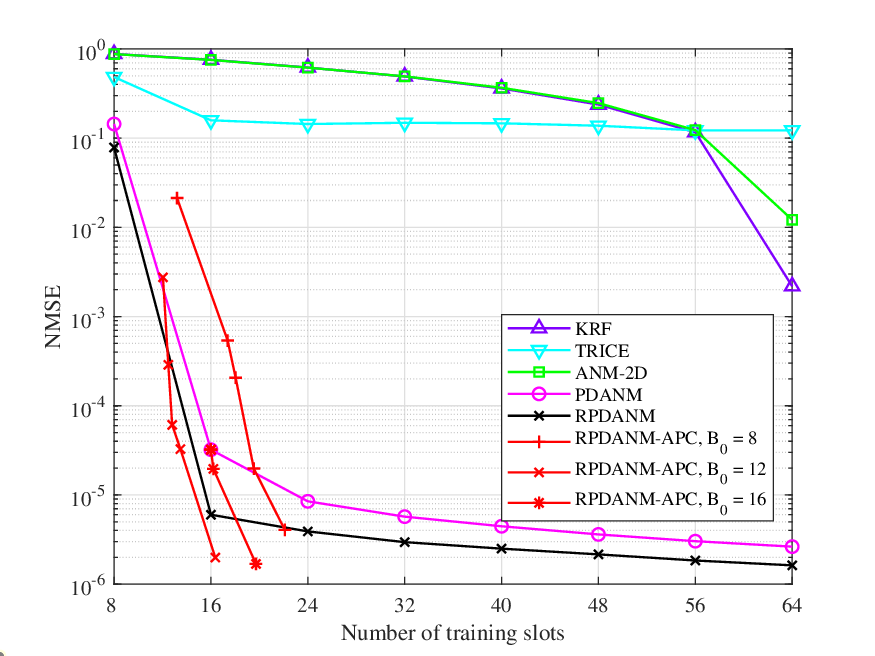}
	\vspace{-3mm}
	\caption{NMSE of different methods versus the number of training slots.}
	\label{fig8}
	\vspace{-4mm}
\end{figure}

In this subsection, we compare the required training slots of different methods to illustrate the ability to save training overhead by the proposed RPDANM-APC approach. Fig. \ref{fig8} depicts the NMSE of different methods versus the number of required training slots, where a 64-element RIS is considered for a clear comparison. Under this setting, ANM-3D cannot be applied due to the exceedingly large size of the PSD matrix in the SDP problem to be solved. For methods except RPDANM-APC, we gradually increase the number of training slots from 8 to 64 and conduct 50 Monte Carlo simulations under each setting to observe their performance. For RPDANM-APC, we select three different initial numbers of training slots to observe its performance, i.e., $B_{0}=8,12,16$ (corresponding to 2, 3, and 4 times the number of effective paths, respectively). Since the number of training slots employed by RPDANM-APC cannot be specified, we assume the actual channel to be known and stop training when the NMSE is smaller than a given threshold ranging from $10^{-1}$ to $10^{-5}$ in this simulation. As above, 50 Monte Carlo simulations are performed on RPDANM-APC for each threshold with average NMSE and average training overhead depicted in Fig. \ref{fig8}. Therefore, this simulation illustrates the number of training slots required for RPDANM-APC to achieve a specified CE accuracy. It is observed in Fig. \ref{fig8} that the performance of KRF drops significantly when the number of training slots $B < N_\mathrm{R}$ since KRF requires $B \geq N_\mathrm{R}$ to guarantee the row orthogonality of $\boldsymbol{\Omega}$, while other methods do not have this limitation. As a CS method, the CE accuracy of TRICE is not greatly affected by $B$. In contrast, the proposed methods achieve high CE accuracy with only a small number of training slots, especially for RPDANM-APC. In particular, RPDANM-APC achieves a higher CE accuracy than PDANM while consuming $30\%$ less training slots compared to the latter.

\subsection{Selection of Initial Number of Training Slots}

\begin{figure}[!t]
	\centering
	\includegraphics[width=5in]{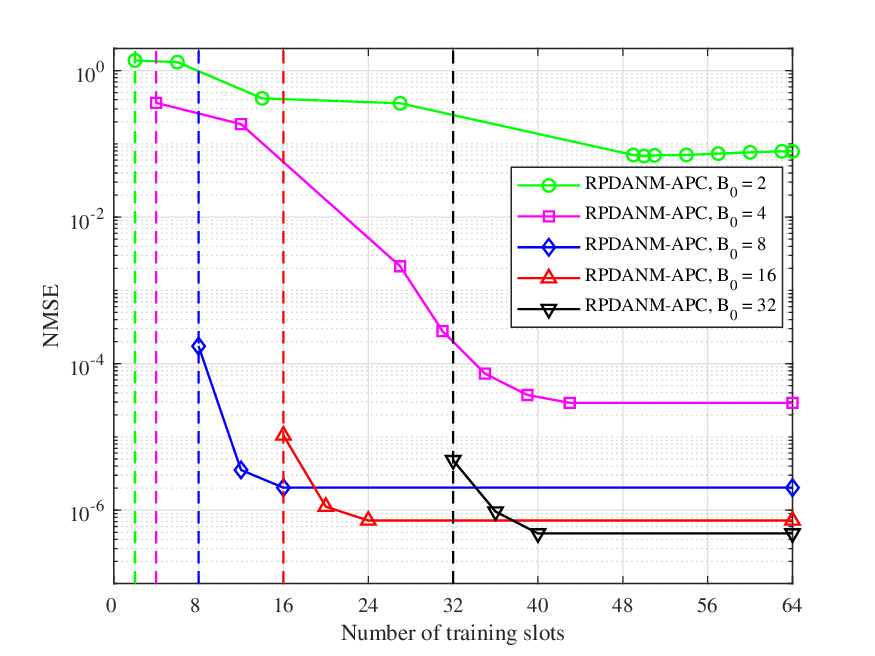}
	\vspace{-3mm}
	\caption{NMSE of RPDANM-APC versus the number of training slots with different $B_{0}$. Dashed lines indicate the initial number of training slots.}
	\label{fig9}
	\vspace{-4mm}
\end{figure}

In this subsection, we evaluate the effect of the initial number of training slots $B_{0}$ on the performance of RPDANM-APC to guide the selection of $B_{0}$. Fig. \ref{fig9} provides an illustrative example about the NMSE of RPDANM-APC versus the number of training slots with different $B_{0}$, where a 64-element RIS is considered. It is observed that more initial training slots lead to a more precise initialization and thus leading to more precise CE ultimately. In particular, RPDANM-APC does not converge when $B_{0}$ is excessively small, e.g., $B_{0}=2$.

\begin{figure}[!t]
	\centering
	\includegraphics[width=5in]{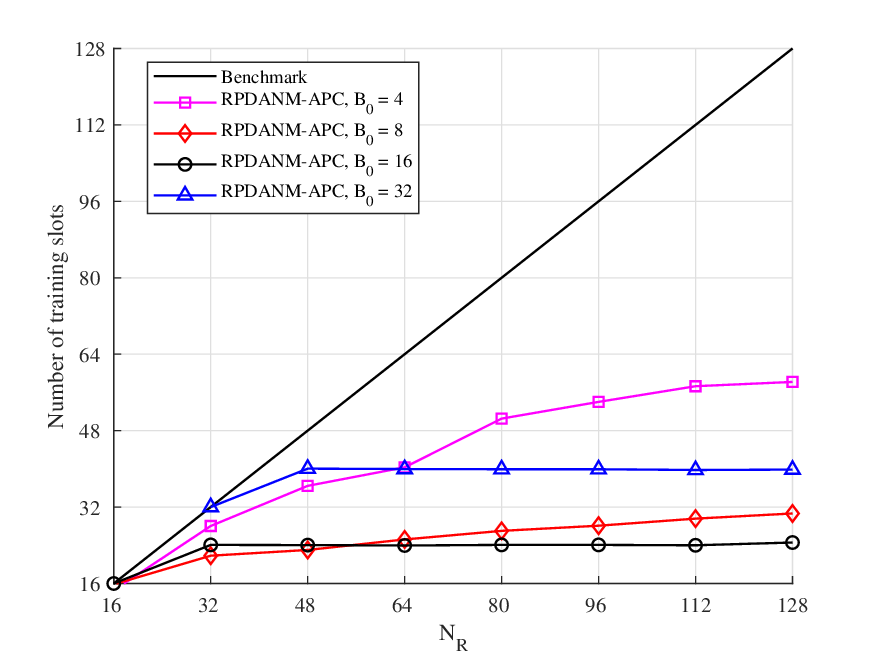}
	\vspace{-3mm}
	\caption{Average number of training slots of RPDANM-APC versus the size of RIS under different $B_{0}$.}
	\label{fig10}
	\vspace{-4mm}
\end{figure}

Furthermore, we investigate the effect of the initial number of training slots $B_{0}$ on the total number of training slots $B$ required for RPDANM-APC via Monte Carlo simulations. Fig. \ref{fig10} illustrates the average number of required training slots under different $B_{0}$ versus the size of RIS. $N_\mathrm{R}$ increases from 16 to 128 and 50 Monte Carlo experiments are conducted under each setting. The black line denotes $B_{\max}$, i.e., the upper bound of number of training slots. It is observed that as $N_\mathrm{R}$ increases, the advantage of RPDANM-APC in saving training overhead is considerable. In particular, when $B_{0}$ is between 8 and 16, i.e., 2 to 4 times of the number of effective paths, the least training slots are required for convergence. On the one hand, when $B_{0}$ is small, e.g., $B_{0} = 4$, imprecise initialization leads to more iterations required for convergence, resulting in more training overhead. On the other hand, when $B_{0}$ is large, e.g., $B_{0} = 32$, too many initial training slots leads to a waste of training overhead. As a result, we recommend setting the initial number of training slots between $2\hat{L}$ and $4\hat{L}$, where $\hat{L}$ is the estimated number of effective paths.

\section{Conclusions} \label{SecCon}

In this paper, a PDANM-based CE framework was proposed for RIS-aided MIMO systems, which reduces the computational complexity while suffering a slight CE accuracy loss compared to state-of-the-art methods. A PDANM-based iterative algorithm called RPDANM was proposed next, which further promotes sparsity by approximately solving a rank minimization problem and thus achieves higher CE accuracy than PDANM. Furthermore, by adaptively adjusting the RIS phases during channel sounding, an iterative CE approach named RPDANM-APC was proposed, which strikes a balance between reduced training overhead and competitive CE performance. Numerical simulations were provided to demonstrate the high computational efficiency, superior CE accuracy, and excellent ability to save training overhead of our proposed methods, especially the RPDANM-APC approach. In future studies, we will extend PDANM to higher-dimensional cases and other parameter estimation scenarios. In addition, efficient algorithms will be developed to further reduce the computational cost of solving large-scale SDP problems.

{\appendices
	\section{Proof of Proposition \ref{Proposition1}} \label{Proposition1Proof}
	For any $\kappa \in \mathbb{C}$ and $\xi \in [-\zeta,\zeta]$ with $\zeta = \min_{l_\mathrm{BR} \in \{1,\cdots,L_\mathrm{BR}\},l_\mathrm{RU} \in \{1,\cdots,L_\mathrm{RU}\}} \big\{ 1-\cos([\boldsymbol{\phi}'_\mathrm{R}]_{l_\mathrm{BR}}), 1+\cos([\boldsymbol{\phi}'_\mathrm{R}]_{l_\mathrm{BR}}), 1-\cos([\boldsymbol{\theta}'_\mathrm{R}]_{l_\mathrm{RU}}), 1+\cos([\boldsymbol{\theta}'_\mathrm{R}]_{l_\mathrm{RU}}) \big\}$, we take $\boldsymbol{\rho}'_\mathrm{BR} = \kappa \boldsymbol{\rho}_\mathrm{BR}$ and $\boldsymbol{\rho}'_\mathrm{RU} = \frac{1}{\kappa} \boldsymbol{\rho}_\mathrm{RU}$, and let $[\boldsymbol{\phi}'_\mathrm{R}]_{l_\mathrm{BR}} = \arccos\big(\cos([\boldsymbol{\phi}_\mathrm{R}]_{l_\mathrm{BR}})+\xi\big)$ and $[\boldsymbol{\theta}'_\mathrm{R}]_{l_\mathrm{RU}} = \arccos\big(\cos([\boldsymbol{\theta}_\mathrm{R}]_{l_\mathrm{RU}})+\xi\big)$ for any $l_\mathrm{BR} \in \{ 1,\cdots,L_\mathrm{BR} \}$ and $l_\mathrm{RU} \in \{ 1,\cdots,L_\mathrm{RU} \}$. Then, we have
	\begin{equation}
		\begin{aligned}
			& \big[ \mathrm{diag} ( \boldsymbol{\rho}'_\mathrm{RU} ) \mathbf{A}^{H}_{N_\mathrm{R}} ( \boldsymbol{\theta}'_\mathrm{R} ) \cdot \mathrm{diag} ( \boldsymbol{\omega}_b ) \cdot \mathbf{A}_{N_\mathrm{R}} ( \boldsymbol{\phi}'_\mathrm{R} ) \mathrm{diag} ( \boldsymbol{\rho}'_\mathrm{BR} ) \big]_{l_\mathrm{RU},l_\mathrm{BR}} \\ = & [ \boldsymbol{\rho}'_\mathrm{RU} ]_{l_\mathrm{RU}} [ \boldsymbol{\rho}'_\mathrm{BR} ]_{l_\mathrm{BR}} \mathbf{a}^{H}_{N_\mathrm{R}} ( [\boldsymbol{\theta}'_\mathrm{R}]_{l_\mathrm{RU}} ) \mathrm{diag} ( \boldsymbol{\omega}_b ) \mathbf{a}_{N_\mathrm{R}} ( [\boldsymbol{\phi}'_\mathrm{R}]_{l_\mathrm{BR}} ) \\ = & [ \boldsymbol{\rho}'_\mathrm{RU} ]_{l_\mathrm{RU}} [ \boldsymbol{\rho}'_\mathrm{BR} ]_{l_\mathrm{BR}} \sum^{N_{\mathrm{R}}}_{n_{\mathrm{R}}=1} [\boldsymbol{\omega}_b]_{n_{\mathrm{R}}} \mathrm{e}^{-i \pi ( n_{\mathrm{R}}-1 ) \cos ([\boldsymbol{\theta}'_\mathrm{R}]_{l_\mathrm{RU}})} \mathrm{e}^{i \pi ( n_{\mathrm{R}}-1 ) \cos ([\boldsymbol{\phi}'_\mathrm{R}]_{l_\mathrm{BR}})} \\ = & \frac{1}{\kappa} [\boldsymbol{\rho}_\mathrm{RU}]_{l_\mathrm{RU}} \cdot \kappa [ \boldsymbol{\rho}_\mathrm{BR} ]_{l_\mathrm{BR}} \cdot \sum^{N_{\mathrm{R}}}_{n_{\mathrm{R}}=1} [\boldsymbol{\omega}_b]_{n_{\mathrm{R}}} \mathrm{e}^{i \pi ( n_{\mathrm{R}}-1 ) \big[ \cos ([\boldsymbol{\phi}_\mathrm{R}]_{l_\mathrm{BR}})+\xi-\cos ([\boldsymbol{\theta}_\mathrm{R}]_{l_\mathrm{RU}})-\xi \big]} \\ = & [\boldsymbol{\rho}_\mathrm{RU}]_{l_\mathrm{RU}} [ \boldsymbol{\rho}_\mathrm{BR} ]_{l_\mathrm{BR}} \sum^{N_{\mathrm{R}}}_{n_{\mathrm{R}}=1} [\boldsymbol{\omega}_b]_{n_{\mathrm{R}}} \mathrm{e}^{-i \pi ( n_{\mathrm{R}}-1 ) \cos ([\boldsymbol{\theta}_\mathrm{R}]_{l_\mathrm{RU}}) } \mathrm{e}^{i \pi ( n_{\mathrm{R}}-1 ) \cos ([\boldsymbol{\phi}_\mathrm{R}]_{l_\mathrm{BR}})} \\ = & \big[ \mathrm{diag} ( \boldsymbol{\rho}_\mathrm{RU} ) \mathbf{A}^{H}_{N_\mathrm{R}} ( \boldsymbol{\theta}_\mathrm{R} ) \cdot \mathrm{diag} ( \boldsymbol{\omega}_b ) \cdot \mathbf{A}_{N_\mathrm{R}} ( \boldsymbol{\phi}_\mathrm{R} ) \mathrm{diag} ( \boldsymbol{\rho}_\mathrm{BR} ) \big]_{l_\mathrm{BR},l_\mathrm{RU}}.
		\end{aligned}
	\end{equation}
	It follows that
	\begin{equation}
		\begin{aligned}
			\mathbf{H}^b_\mathrm{BU} & = \mathbf{A}_{N_\mathrm{U}} ( \boldsymbol{\phi}_\mathrm{U} ) \mathrm{diag} ( \boldsymbol{\rho}_\mathrm{RU} ) \mathbf{A}^{H}_{N_\mathrm{R}} ( \boldsymbol{\theta}_\mathrm{R} ) \cdot \mathrm{diag} ( \boldsymbol{\omega}_b ) \cdot \mathbf{A}_{N_\mathrm{R}} ( \boldsymbol{\phi}_\mathrm{R} ) \mathrm{diag} ( \boldsymbol{\rho}_\mathrm{BR} ) \mathbf{A}^{H}_{N_\mathrm{B}} ( \boldsymbol{\theta}_\mathrm{B} ) \\ 
			& = \mathbf{A}_{N_\mathrm{U}} ( \boldsymbol{\phi}_\mathrm{U} ) \mathrm{diag} ( \boldsymbol{\rho}'_\mathrm{RU} ) \mathbf{A}^{H}_{N_\mathrm{R}} ( \boldsymbol{\theta}'_\mathrm{R} ) \cdot \mathrm{diag} ( \boldsymbol{\omega}_b ) \cdot \mathbf{A}_{N_\mathrm{R}} ( \boldsymbol{\phi}'_\mathrm{R} ) \mathrm{diag} ( \boldsymbol{\rho}'_\mathrm{BR} ) \mathbf{A}^{H}_{N_\mathrm{B}} ( \boldsymbol{\theta}_\mathrm{B} ).
		\end{aligned}
	\end{equation}
	
	\section{Proof of Theorem \ref{Theorem1}} \label{Theorem1Proof}
	We first show that $\mathrm{SDP}(\mathbf{H}) \leq \| \mathbf{H} \|_{\mathcal{A}}$. Assume that
	\begin{equation} \label{HAtoDec}
		\begin{aligned}
			\mathbf{H} & = \sum_{l_\mathrm{BU}} \sum_{l_\mathrm{R}} [\mathbf{P}]_{l_\mathrm{BU},l_\mathrm{R}} \big[ \mathbf{a}_{N_\mathrm{B}} ( [\boldsymbol{\theta}]_{l_\mathrm{BU}} ) \otimes \mathbf{a}_{N_\mathrm{U}} ( [\boldsymbol{\phi}]_{l_\mathrm{BU}} ) \big] \mathbf{a}^{H}_{N_\mathrm{R}} ( \psi_{l_\mathrm{R}} ) \\ & = \big[ \mathbf{A}_{N_\mathrm{B}} ( \boldsymbol{\theta} ) \diamond \mathbf{A}_{N_\mathrm{U}} ( \boldsymbol{\phi} ) \big] \mathbf{P} \mathbf{A}^{H}_{N_\mathrm{R}} ( \boldsymbol{\psi} )
		\end{aligned}
	\end{equation}
	is a partially decoupled atomic decomposition of the effective channel, where $\mathbf{P}$ is not necessarily a diagonal matrix since there may be repeated entries in $\boldsymbol{\theta}$ and $\boldsymbol{\phi}$ since different partially decoupled atoms may overlap in some dimensions. We take
	\begin{equation}
		\begin{aligned}
			\mathcal{T}_{[N_\mathrm{B},N_\mathrm{U}]}(\mathbf{T}) & = \sum_{l_\mathrm{BU}} \sum_{l_\mathrm{R}} \left | [\mathbf{P}]_{l_\mathrm{BU},l_\mathrm{R}} \right | \big[ \mathbf{a}_{N_\mathrm{B}} ( [\boldsymbol{\theta}]_{l_\mathrm{BU}} ) \otimes \mathbf{a}_{N_\mathrm{U}} ( [\boldsymbol{\phi}]_{l_\mathrm{BU}} ) \big] \big[ \mathbf{a}_{N_\mathrm{B}} ( [\boldsymbol{\theta}]_{l_\mathrm{BU}} ) \otimes \mathbf{a}_{N_\mathrm{U}} ( [\boldsymbol{\phi}]_{l_\mathrm{BU}} ) \big]^{H}
		\end{aligned}
	\end{equation}
	and
	\begin{equation}
		\mathcal{T}_{N_\mathrm{R}}(\mathbf{t}) = \sum_{l_\mathrm{BU}} \sum_{l_\mathrm{R}} \left | [\mathbf{P}]_{l_\mathrm{BU},l_\mathrm{R}} \right | \mathbf{a}_{N_\mathrm{R}} ( \psi_{l_\mathrm{R}} ) \mathbf{a}^{H}_{N_\mathrm{R}} ( \psi_{l_\mathrm{R}} ).
	\end{equation}
	It follows that
	\begin{equation}
		\begin{aligned}
			\begin{bmatrix}
				\mathcal{T}_{N_\mathrm{R}}(\mathbf{t}) & \hspace{-3mm} \mathbf{H}^{H}\\
				\mathbf{H} & \hspace{-3mm} \mathcal{T}_{[N_\mathrm{B},N_\mathrm{U}]}(\mathbf{T})
			\end{bmatrix} & = \sum_{l_\mathrm{BU}} \sum_{l_\mathrm{R}} \left | [\mathbf{P}]_{l_\mathrm{BU},l_\mathrm{R}} \right | \\ & \quad \cdot
			\begin{bmatrix}
				\frac{\overline{[\mathbf{P}]_{l_\mathrm{BU},l_\mathrm{R}}}}{\left | [\mathbf{P}]_{l_\mathrm{BU},l_\mathrm{R}} \right |} \mathbf{a}_{N_\mathrm{R}} ( \psi_{l_\mathrm{R}} ) \\ \mathbf{a}_{N_\mathrm{B}} ( [\boldsymbol{\theta}]_{l_\mathrm{BU}} ) \otimes \mathbf{a}_{N_\mathrm{U}} ( [\boldsymbol{\phi}]_{l_\mathrm{BU}} )
			\end{bmatrix}
			\begin{bmatrix}
				\frac{\overline{[\mathbf{P}]_{l_\mathrm{BU},l_\mathrm{R}}}}{\left | [\mathbf{P}]_{l_\mathrm{BU},l_\mathrm{R}} \right |} \mathbf{a}_{N_\mathrm{R}} ( \psi_{l_\mathrm{R}} ) \\ \mathbf{a}_{N_\mathrm{B}} ( [\boldsymbol{\theta}]_{l_\mathrm{BU}} ) \otimes \mathbf{a}_{N_\mathrm{U}} ( [\boldsymbol{\phi}]_{l_\mathrm{BU}} )
			\end{bmatrix}^{H} \succeq \mathbf{0},
		\end{aligned}
	\end{equation}
	and thus
	\begin{equation} \label{SDPleq}
		\begin{aligned}
			\mathrm{SDP}(\mathbf{H}) & \leq \frac{1}{2N_\mathrm{R}} \mathrm{tr} \big(\mathcal{T}_{N_\mathrm{R}}(\mathbf{t})\big) + \frac{1}{2N_\mathrm{B}N_\mathrm{U}} \mathrm{tr} \big(\mathcal{T}_{[N_\mathrm{B},N_\mathrm{U}]}(\mathbf{T})\big) \\ & = \frac{1}{2N_\mathrm{R}} \sum_{l_\mathrm{BU}} \sum_{l_\mathrm{R}} \left | [\mathbf{P}]_{l_\mathrm{BU},l_\mathrm{R}} \right | \mathbf{a}^{H}_{N_\mathrm{R}} ( \psi_{l_\mathrm{R}} ) \mathbf{a}_{N_\mathrm{R}} ( \psi_{l_\mathrm{R}} ) + \frac{1}{2N_\mathrm{B}N_\mathrm{U}} \sum_{l_\mathrm{BU}} \sum_{l_\mathrm{R}} \left | [\mathbf{P}]_{l_\mathrm{BU},l_\mathrm{R}} \right | \\ & \quad \cdot \big[ \mathbf{a}_{N_\mathrm{B}} ( [\boldsymbol{\theta}]_{l_\mathrm{BU}} ) \otimes \mathbf{a}_{N_\mathrm{U}} ( [\boldsymbol{\phi}]_{l_\mathrm{BU}} ) \big]^{H} \big[ \mathbf{a}_{N_\mathrm{B}} ( [\boldsymbol{\theta}]_{l_\mathrm{BU}} ) \otimes \mathbf{a}_{N_\mathrm{U}} ( [\boldsymbol{\phi}]_{l_\mathrm{BU}} ) \big] \\ & = \sum_{l_\mathrm{BU}} \sum_{l_\mathrm{R}} \left | [\mathbf{P}]_{l_\mathrm{BU},l_\mathrm{R}} \right |.
		\end{aligned}
	\end{equation}
	Since \eqref{SDPleq} holds for any partially decoupled atomic decomposition of $\mathbf{H}$, we have that $\mathrm{SDP}(\mathbf{H}) \leq \| \mathbf{H} \|_{\mathcal{A}}$.
	
	We next show that $\mathrm{SDP}(\mathbf{H}) = \| \mathbf{H} \|_{\mathcal{A}}$ under the three conditions in Theorem \ref{Theorem1}. Due to the column inclusion property \cite{horn2012matrix} of PSD matrices, we have that $\mathrm{col} \big( \mathbf{H}^{H} \big) \in \mathrm{col} \big( \mathcal{T}_{N_\mathrm{R}}(\mathbf{t}^{*}) \big)$ and $\mathrm{col} \big( \mathbf{H} \big) \in \mathrm{col} \big( \mathcal{T}_{[N_\mathrm{B},N_\mathrm{U}]}(\mathbf{T}^{*}) \big)$. It follows from \eqref{VanDec1} and \eqref{VanDec2} that there exists a matrix $\mathbf{C} \in \mathbb{C}^{L_\mathrm{BU}^{*} \times L_\mathrm{R}^{*}}$ such that $\mathbf{H}$ admits the following partially decoupled atomic decomposition:
	\begin{equation} \label{HAtoDec*}
		\mathbf{H} = \big[ \mathbf{A}_{N_\mathrm{B}} ( \boldsymbol{\theta}_\mathrm{B}^{*} ) \diamond \mathbf{A}_{N_\mathrm{U}} ( \boldsymbol{\phi}_\mathrm{U}^{*} ) \big] \mathbf{C} \mathbf{A}^{H}_{N_\mathrm{R}} ( \boldsymbol{\psi}_\mathrm{R}^{*} ).
	\end{equation}
	
	We first assume that there is exactly one non-zero element in each row of $\mathbf{C}$, denoted as $[\mathbf{C}]_{i,k_{i}}$ in the $i$-th row. It follows from the Schur complement theorem \cite{zhang2006schur} that $\mathcal{T}_{N_\mathrm{R}}(\mathbf{t}^{*}) \succeq \mathbf{H}^{H} {\mathcal{T}_{[N_\mathrm{B},N_\mathrm{U}]}(\mathbf{T}^{*})}^{\dagger} \mathbf{H}$, i.e.,
	\begin{equation} \label{Schsucceq}
		\begin{aligned}
			& \mathbf{A}_{N_\mathrm{R}} ( \boldsymbol{\psi}_\mathrm{R}^{*} ) \mathrm{diag} (\mathbf{p}_\mathrm{R}^{*}) \mathbf{A}^{H}_{N_\mathrm{B}} ( \boldsymbol{\psi}_\mathrm{R}^{*} ) \\ \succeq & \mathbf{A}_{N_\mathrm{R}} ( \boldsymbol{\psi}_\mathrm{R}^{*} ) \mathbf{C}^{H} \big[ \mathbf{A}_{N_\mathrm{B}} ( \boldsymbol{\theta}_\mathrm{B}^{*} ) \diamond \mathbf{A}_{N_\mathrm{U}} ( \boldsymbol{\phi}_\mathrm{U}^{*} ) \big]^{H} {\big[ \mathbf{A}_{N_\mathrm{B}} ( \boldsymbol{\theta}_\mathrm{B}^{*} ) \diamond \mathbf{A}_{N_\mathrm{U}} ( \boldsymbol{\phi}_\mathrm{U}^{*} ) \big]^{H}}^{\dagger} \\ & \cdot {\mathrm{diag} (\mathbf{p}_\mathrm{BU}^{*})}^{\dagger} \big[ \mathbf{A}_{N_\mathrm{B}} ( \boldsymbol{\theta}_\mathrm{B}^{*} ) \diamond \mathbf{A}_{N_\mathrm{U}} ( \boldsymbol{\phi}_\mathrm{U}^{*} ) \big]^{\dagger} \big[ \mathbf{A}_{N_\mathrm{B}} ( \boldsymbol{\theta}_\mathrm{B}^{*} ) \diamond \mathbf{A}_{N_\mathrm{U}} ( \boldsymbol{\phi}_\mathrm{U}^{*} ) \big] \mathbf{C} \mathbf{A}^{H}_{N_\mathrm{R}} ( \boldsymbol{\psi}_\mathrm{R}^{*} ) \\ = & \mathbf{A}_{N_\mathrm{R}} ( \boldsymbol{\psi}_\mathrm{R}^{*} ) \mathbf{C}^{H} {\mathrm{diag} (\mathbf{p}_\mathrm{BU}^{*})}^{\dagger} \mathbf{C} \mathbf{A}^{H}_{N_\mathrm{R}} ( \boldsymbol{\psi}_\mathrm{R}^{*} ).
		\end{aligned}
	\end{equation}
	From \eqref{Schsucceq}, we obtain
	\begin{equation} \label{diagsucceq1}
		\mathrm{diag} (\mathbf{p}_\mathrm{R}^{*}) \succeq \mathbf{C}^{H} {\mathrm{diag} (\mathbf{p}_\mathrm{BU}^{*})}^{\dagger} \mathbf{C},
	\end{equation}
	which yields
	\begin{equation} \label{diagsucceq2}
		[\mathbf{p}_\mathrm{R}^{*}]_{l_\mathrm{R}} \geq \sum_{l_\mathrm{BU}=1}^{L_\mathrm{BU}^{*}} \frac{\left| [\mathbf{C}]_{l_\mathrm{BU},l_\mathrm{R}} \right|^{2}}{[\mathbf{p}_\mathrm{BU}^{*}]_{l_\mathrm{BU}}}, l_\mathrm{R}=1,\cdots,L_\mathrm{R}^{*}
	\end{equation}
	since $\mathbf{C}^{H} {\mathrm{diag} (\mathbf{p}_\mathrm{BU}^{*})}^{\dagger} \mathbf{C}$ is a diagonal matrix. Then, we have
	\begin{equation} \label{SDPgeq1}
		\begin{aligned}
			\mathrm{SDP}(\mathbf{H}) & = \frac{1}{2N_\mathrm{R}} \mathrm{tr} \big(\mathcal{T}_{N_\mathrm{R}}(\mathbf{t}^{*})\big) + \frac{1}{2N_\mathrm{B}N_\mathrm{U}} \mathrm{tr} \big(\mathcal{T}_{[N_\mathrm{B},N_\mathrm{U}]}(\mathbf{T}^{*})\big) \\ & = \frac{1}{2N_\mathrm{R}} \sum_{l_\mathrm{R}=1}^{L_\mathrm{R}^{*}} [\mathbf{p}_\mathrm{R}^{*}]_{l_\mathrm{R}} \mathbf{a}^{H}_{N_\mathrm{R}} ( [\boldsymbol{\psi}_\mathrm{R}^{*}]_{l_{R}} ) \mathbf{a}_{N_\mathrm{R}} ( [\boldsymbol{\psi}_\mathrm{R}^{*}]_{l_{R}} ) + \frac{1}{2N_\mathrm{B}N_\mathrm{U}} \sum_{l_\mathrm{BU}=1}^{L_\mathrm{BU}^{*}} [\mathbf{p}_\mathrm{BU}^{*}]_{l_\mathrm{BU}} \\ & \quad \cdot \big[ \mathbf{a}_{N_\mathrm{B}} ( [\boldsymbol{\theta}_\mathrm{B}^{*}]_{l_{BU}} ) \otimes \mathbf{a}_{N_\mathrm{U}} ( [\boldsymbol{\phi}_\mathrm{U}^{*}]_{l_{BU}} ) \big]^{H} \big[ \mathbf{a}_{N_\mathrm{B}} ( [\boldsymbol{\theta}_\mathrm{B}^{*}]_{l_{BU}} ) \otimes \mathbf{a}_{N_\mathrm{U}} ( [\boldsymbol{\phi}_\mathrm{U}^{*}]_{l_{BU}} ) \big] \\ & = \frac{1}{2} \sum_{l_\mathrm{R}=1}^{L_\mathrm{R}^{*}} [\mathbf{p}_\mathrm{R}^{*}]_{l_\mathrm{R}} + \frac{1}{2} \sum_{l_\mathrm{BU}=1}^{L_\mathrm{BU}^{*}} [\mathbf{p}_\mathrm{BU}^{*}]_{l_\mathrm{BU}} \\ & \geq \frac{1}{2} \sum_{l_\mathrm{R}=1}^{L_\mathrm{R}^{*}} \sum_{l_\mathrm{BU}=1}^{L_\mathrm{BU}^{*}} \frac{\left| [\mathbf{C}]_{l_\mathrm{BU},l_\mathrm{R}} \right|^{2}}{[\mathbf{p}_\mathrm{BU}^{*}]_{l_\mathrm{BU}}} + \frac{1}{2} \sum_{l_\mathrm{BU}=1}^{L_\mathrm{BU}^{*}} [\mathbf{p}_\mathrm{BU}^{*}]_{l_\mathrm{BU}} \\ & = \frac{1}{2} \sum_{l_\mathrm{BU}=1}^{L_\mathrm{BU}^{*}} \frac{\left| [\mathbf{C}]_{l_\mathrm{BU},k_{l_\mathrm{BU}}} \right|^{2}}{[\mathbf{p}_\mathrm{BU}^{*}]_{l_\mathrm{BU}}} + \frac{1}{2} \sum_{l_\mathrm{BU}=1}^{L_\mathrm{BU}^{*}} [\mathbf{p}_\mathrm{BU}^{*}]_{l_\mathrm{BU}} \\ & \geq \sum_{l_\mathrm{BU}=1}^{L_\mathrm{BU}^{*}} \left| [\mathbf{C}]_{l_\mathrm{BU},k_{l_\mathrm{BU}}} \right| \geq \| \mathbf{H} \|_{\mathcal{A}}.
		\end{aligned}
	\end{equation}
	Combining \eqref{SDPgeq1} and \eqref{SDPleq}, we draw the conclusion that $\mathrm{SDP}(\mathbf{H}) = \| \mathbf{H} \|_{\mathcal{A}}$.
	
	Similarly, if there is exactly one non-zero element in each column of $\mathbf{C}$, denoted as $[\mathbf{C}]_{k_{j},j}$ in the $j$-th column, it follows that $\mathcal{T}_{[N_\mathrm{B},N_\mathrm{U}]}(\mathbf{T}^{*}) \succeq \mathbf{H} {\mathcal{T}_{N_\mathrm{R}}(\mathbf{t}^{*})}^{\dagger} \mathbf{H}^{H}$. Similarly to the derivations above, we obtain that
	\begin{equation} \label{diagsucceq3}
		\mathrm{diag} (\mathbf{p}_\mathrm{BU}^{*}) \succeq \mathbf{C} {\mathrm{diag} (\mathbf{p}_\mathrm{R}^{*})}^{\dagger} \mathbf{C}^{H},
	\end{equation}
	or equivalently,
	\begin{equation} \label{diagsucceq4}
		[\mathbf{p}_\mathrm{BU}^{*}]_{l_\mathrm{BU}} \geq \sum_{l_\mathrm{R}=1}^{L_\mathrm{R}^{*}} \frac{\left| [\mathbf{C}]_{l_\mathrm{BU},l_\mathrm{R}} \right|^{2}}{[\mathbf{p}_\mathrm{R}^{*}]_{l_\mathrm{R}}}, l_\mathrm{BU}=1,\cdots,L_\mathrm{BU}^{*},
	\end{equation}
	since $\mathbf{C} {\mathrm{diag} (\mathbf{p}_\mathrm{R}^{*})}^{\dagger} \mathbf{C}^{H}$ is a diagonal matrix. Then, we have
	\begin{equation} \label{SDPgeq2}
		\begin{aligned}
			\mathrm{SDP}(\mathbf{H}) & = \frac{1}{2} \sum_{l_\mathrm{R}=1}^{L_\mathrm{R}^{*}} [\mathbf{p}_\mathrm{R}^{*}]_{l_\mathrm{R}} + \frac{1}{2} \sum_{l_\mathrm{BU}=1}^{L_\mathrm{BU}^{*}} [\mathbf{p}_\mathrm{BU}^{*}]_{l_\mathrm{BU}} \\ & \geq \frac{1}{2} \sum_{l_\mathrm{R}=1}^{L_\mathrm{R}^{*}} [\mathbf{p}_\mathrm{R}^{*}]_{l_\mathrm{R}} + \frac{1}{2} \sum_{l_\mathrm{BU}=1}^{L_\mathrm{BU}^{*}} \sum_{l_\mathrm{R}=1}^{L_\mathrm{R}^{*}} \frac{\left| [\mathbf{C}]_{l_\mathrm{BU},l_\mathrm{R}} \right|^{2}}{[\mathbf{p}_\mathrm{R}^{*}]_{l_\mathrm{R}}} \\ & = \frac{1}{2} \sum_{l_\mathrm{R}=1}^{L_\mathrm{R}^{*}} [\mathbf{p}_\mathrm{R}^{*}]_{l_\mathrm{R}} + \frac{1}{2} \sum_{l_\mathrm{R}=1}^{L_\mathrm{R}^{*}} \frac{\left| [\mathbf{C}]_{k_{l_\mathrm{R}},l_\mathrm{R}} \right|^{2}}{[\mathbf{p}_\mathrm{R}^{*}]_{l_\mathrm{R}}} \\ & \geq \sum_{l_\mathrm{R}=1}^{L_\mathrm{R}^{*}} \left| [\mathbf{C}]_{k_{l_\mathrm{R}},l_\mathrm{R}} \right| \geq \| \mathbf{H} \|_{\mathcal{A}}.
		\end{aligned}
	\end{equation}
	Combining \eqref{SDPgeq2} and \eqref{SDPleq}, we obtain that $\mathrm{SDP}(\mathbf{H}) = \| \mathbf{H} \|_{\mathcal{A}}$ and complete the proof.
	
	\section{Proof of Theorem \ref{Theorem2}} \label{Theorem2Proof}
	It follows from the proof of Theorem \ref{Theorem1} that $\mathrm{SDP}(\mathbf{H}) \leq \| \mathbf{H} \|_{\mathcal{A}}$ holds. Next, we will show that $\mathrm{SDP}(\mathbf{H}) \geq \| \mathbf{H} \|_{\mathcal{A}}$. Define the one-dimensional (1D) atomic set as
	\begin{equation} \label{AS1D}
		\begin{aligned}
			\mathcal{A}_{1D} & = \Big\{ \mathbf{b} \mathbf{a}^{H}_{N_\mathrm{R}} ( \psi ) \in \mathbb{C}^{N_\mathrm{B}N_\mathrm{U} \times N_\mathrm{R}} : \psi \in [0,\pi], \| \mathbf{b} \|_{2} = 1 \Big\} 
		\end{aligned}
	\end{equation}
	and the corresponding 1D atomic norm as
	\begin{equation} \label{AN1D}
		\begin{aligned}
			\| \mathbf{H} \|_{\mathcal{A}_{1D}} & = \inf \Big\{ \sum_{l} | \rho_{l} |:\mathbf{H} =\sum_{l} \rho_{l} \mathbf{b}_{l} \mathbf{a}^{H}_{N_\mathrm{R}} ( \psi_{l} ) : \psi_{l} \in [0,\pi], \| \mathbf{b}_{l} \|_{2} = 1 \Big\},
		\end{aligned}
	\end{equation}
	then the effective channel $\mathbf{H}$ admits a 1D atomic decomposition in the 1D atomic set as
	\begin{equation}
		\mathbf{H} = \sum_{l_\mathrm{BR}=1}^{L_\mathrm{BR}} \sum_{l_\mathrm{RU}=1}^{L_\mathrm{RU}} \sqrt{N_\mathrm{B}N_\mathrm{U}} [\boldsymbol{\rho}_\mathrm{BU}]_{l_\mathrm{BU}} \big[ \frac{1}{\sqrt{N_\mathrm{B}N_\mathrm{U}}} \mathbf{a}_{N_\mathrm{B}} ( [-\boldsymbol{\theta}_\mathrm{B}]_{l_\mathrm{BR}} ) \otimes \mathbf{a}_{N_\mathrm{U}} ( [\boldsymbol{\phi}_\mathrm{U}]_{l_\mathrm{RU}} ) \big] \mathbf{a}^{H}_{N_\mathrm{R}} ( [\boldsymbol{\psi}_\mathrm{R}]_{l_\mathrm{BU}} ).
	\end{equation}
	Denote $\mathrm{SDP}_{1D}(\mathbf{H})$ as the optimal value of the following SDP problem:
	\begin{equation} \label{AN1Dopt}
		\min_{\mathbf{t},\mathbf{Q}} \frac{1}{2N_\mathrm{R}} \mathrm{tr} \big(\mathcal{T}_{N_\mathrm{R}}(\mathbf{t})\big) + \frac{1}{2N_\mathrm{B}N_\mathrm{U}} \mathrm{tr} \big(\mathbf{Q}\big), \ \mathrm{s.t.}
		\begin{bmatrix}
			\mathcal{T}_{N_\mathrm{R}}(\mathbf{t}) & \mathbf{H}^{H}\\
			\mathbf{H} & \mathbf{Q}
		\end{bmatrix}
		\succeq \mathbf{0}.
	\end{equation}
	It is noted that \eqref{PDANopt} essentially imposes an additional constraint that $\mathbf{Q}$ is a 2-level Toeplitz matrix on \eqref{AN1Dopt} and thus
	\begin{equation} \label{th2eq1}
		\mathrm{SDP}(\mathbf{H}) \geq \mathrm{SDP}_{1D}(\mathbf{H}).
	\end{equation}
	It follows from \cite[Theorem 3]{yang2016exact} that
	\begin{equation} \label{th2eq2}
		\mathrm{SDP}_{1D}(\mathbf{H}) = \frac{1}{\sqrt{N_\mathrm{B}N_\mathrm{U}}}\| \mathbf{H} \|_{\mathcal{A}_{1D}}.
	\end{equation}
	Since $\Delta_{\boldsymbol{\psi}_\mathrm{R}} > \frac{4}{N_\mathrm{R}}$, it is proved in \cite[Theorem 4]{yang2016exact} that
	\begin{equation} \label{th2eq3}
		\frac{1}{\sqrt{N_\mathrm{B}N_\mathrm{U}}}\| \mathbf{H} \|_{\mathcal{A}_{1D}} = \sum_{l_\mathrm{BR}=1}^{L_\mathrm{BR}} \sum_{l_\mathrm{RU}=1}^{L_\mathrm{RU}} |[\boldsymbol{\rho}_\mathrm{BU}]_{l_\mathrm{BU}}|
	\end{equation}
	with $l_\mathrm{BU} = (l_{\mathrm{BR}}-1)L_{\mathrm{RU}}+l_\mathrm{RU}$. Furthermore, we have that
	\begin{equation} \label{th2eq4}
		\sum_{l_\mathrm{BR}=1}^{L_\mathrm{BR}} \sum_{l_\mathrm{RU}=1}^{L_\mathrm{RU}} |[\boldsymbol{\rho}_\mathrm{BU}]_{l_\mathrm{BU}}| \geq \| \mathbf{H} \|_{\mathcal{A}},
	\end{equation}
	since \eqref{H} is a partially decoupled atomic decomposition of $\mathbf{H}$. Combining \eqref{th2eq1}, \eqref{th2eq2}, \eqref{th2eq3}, and \eqref{th2eq4}, we have
	\begin{equation}
		\mathrm{SDP}(\mathbf{H}) \geq \mathrm{SDP}_{1D}(\mathbf{H}) = \frac{1}{\sqrt{N_\mathrm{B}N_\mathrm{U}}}\| \mathbf{H} \|_{\mathcal{A}_{1D}} = \sum_{l_\mathrm{BR}=1}^{L_\mathrm{BR}} \sum_{l_\mathrm{RU}=1}^{L_\mathrm{RU}} |[\boldsymbol{\rho}_\mathrm{BU}]_{l_\mathrm{BU}}| \geq \| \mathbf{H} \|_{\mathcal{A}}
	\end{equation}
	and complete the proof.
	
	\section{Proof of Theorem \ref{Theorem3}} \label{Theorem3Proof}
	The proof is a generalization of that of Theorem \ref{Theorem1}. We first show that $\mathrm{SDP}_{\mathbf{W}_{\mathrm{R}},\mathbf{W}_{\mathrm{BU}}}(\mathbf{H}) \leq \| \mathbf{H} \|_{\mathcal{A}^{w_{\mathrm{R}},w_{\mathrm{BU}}}}$. For any partially decoupled atomic decomposition of $\mathbf{H}$ as in \eqref{HAtoDec}, we take
	\begin{equation}
		\begin{aligned}
			\mathcal{T}_{[N_\mathrm{B},N_\mathrm{U}]}(\mathbf{T}) & = \sum_{l_\mathrm{BU}} \sum_{l_\mathrm{R}} \frac{ w_{\mathrm{BU}}([\boldsymbol{\theta}]_{l_\mathrm{BU}},[\boldsymbol{\phi}]_{l_\mathrm{BU}})^{\frac{1}{2}} }{ w_{\mathrm{R}}(\psi_{l_\mathrm{R}})^{\frac{1}{2}} } \left | [\mathbf{P}]_{l_\mathrm{BU},l_\mathrm{R}} \right | \\ & \quad \cdot \big[ \mathbf{a}_{N_\mathrm{B}} ( [\boldsymbol{\theta}]_{l_\mathrm{BU}} ) \otimes \mathbf{a}_{N_\mathrm{U}} ( [\boldsymbol{\phi}]_{l_\mathrm{BU}} ) \big] \big[ \mathbf{a}_{N_\mathrm{B}} ( [\boldsymbol{\theta}]_{l_\mathrm{BU}} ) \otimes \mathbf{a}_{N_\mathrm{U}} ( [\boldsymbol{\phi}]_{l_\mathrm{BU}} ) \big]^{H}
		\end{aligned}
	\end{equation}
	and
	\begin{equation}
		\mathcal{T}_{N_\mathrm{R}}(\mathbf{t}) = \sum_{l_\mathrm{BU}} \sum_{l_\mathrm{R}} \frac{ w_{\mathrm{R}}(\psi_{l_\mathrm{R}})^{\frac{1}{2}} }{ w_{\mathrm{BU}}([\boldsymbol{\theta}]_{l_\mathrm{BU}},[\boldsymbol{\phi}]_{l_\mathrm{BU}})^{\frac{1}{2}} } \left | [\mathbf{P}]_{l_\mathrm{BU},l_\mathrm{R}} \right | \mathbf{a}_{N_\mathrm{R}} ( \psi_{l_\mathrm{R}} ) \mathbf{a}^{H}_{N_\mathrm{R}} ( \psi_{l_\mathrm{R}} ).
	\end{equation}
	It follows that
	\begin{equation}
		\begin{aligned}
			\begin{bmatrix}
				\mathcal{T}_{N_\mathrm{R}}(\mathbf{t}) & \hspace{-3mm} \mathbf{H}^{H}\\
				\mathbf{H} & \hspace{-3mm} \mathcal{T}_{[N_\mathrm{B},N_\mathrm{U}]}(\mathbf{T})
			\end{bmatrix} & = \sum_{l_\mathrm{BU}} \sum_{l_\mathrm{R}} \left | [\mathbf{P}]_{l_\mathrm{BU},l_\mathrm{R}} \right |
			\begin{bmatrix}
				\frac{\overline{[\mathbf{P}]_{l_\mathrm{BU},l_\mathrm{R}}}}{\left | [\mathbf{P}]_{l_\mathrm{BU},l_\mathrm{R}} \right |} \frac{ w_{\mathrm{R}}(\psi_{l_\mathrm{R}})^{\frac{1}{4}} }{ w_{\mathrm{BU}}([\boldsymbol{\theta}]_{l_\mathrm{BU}},[\boldsymbol{\phi}]_{l_\mathrm{BU}})^{\frac{1}{4}} } \mathbf{a}_{N_\mathrm{R}} ( \psi_{l_\mathrm{R}} ) \\ \frac{ w_{\mathrm{BU}}([\boldsymbol{\theta}]_{l_\mathrm{BU}},[\boldsymbol{\phi}]_{l_\mathrm{BU}})^{\frac{1}{4}} }{ w_{\mathrm{R}}(\psi_{l_\mathrm{R}})^{\frac{1}{4}} } \mathbf{a}_{N_\mathrm{B}} ( [\boldsymbol{\theta}]_{l_\mathrm{BU}} ) \otimes \mathbf{a}_{N_\mathrm{U}} ( [\boldsymbol{\phi}]_{l_\mathrm{BU}} )
			\end{bmatrix}
			\\ & \quad \cdot
			\begin{bmatrix}
				\frac{\overline{[\mathbf{P}]_{l_\mathrm{BU},l_\mathrm{R}}}}{\left | [\mathbf{P}]_{l_\mathrm{BU},l_\mathrm{R}} \right |} \frac{ w_{\mathrm{R}}(\psi_{l_\mathrm{R}})^{\frac{1}{4}} }{ w_{\mathrm{BU}}([\boldsymbol{\theta}]_{l_\mathrm{BU}},[\boldsymbol{\phi}]_{l_\mathrm{BU}})^{\frac{1}{4}} } \mathbf{a}_{N_\mathrm{R}} ( \psi_{l_\mathrm{R}} ) \\ \frac{ w_{\mathrm{BU}}([\boldsymbol{\theta}]_{l_\mathrm{BU}},[\boldsymbol{\phi}]_{l_\mathrm{BU}})^{\frac{1}{4}} }{ w_{\mathrm{R}}(\psi_{l_\mathrm{R}})^{\frac{1}{4}} } \mathbf{a}_{N_\mathrm{B}} ( [\boldsymbol{\theta}]_{l_\mathrm{BU}} ) \otimes \mathbf{a}_{N_\mathrm{U}} ( [\boldsymbol{\phi}]_{l_\mathrm{BU}} )
			\end{bmatrix}^{H} \succeq \mathbf{0},
		\end{aligned}
	\end{equation}
	thus we have
	\begin{equation} \label{SDPWleq}
		\begin{aligned}
			\mathrm{SDP}_{\mathbf{W}_{\mathrm{R}},\mathbf{W}_{\mathrm{BU}}}(\mathbf{H}) & \leq \frac{1}{2} \mathrm{tr} \big(\mathbf{W}_{\mathrm{R}}\mathcal{T}_{N_\mathrm{R}}(\mathbf{t})\big) + \frac{1}{2} \mathrm{tr} \big(\mathbf{W}_{\mathrm{BU}}\mathcal{T}_{[N_\mathrm{B},N_\mathrm{U}]}(\mathbf{T})\big) \\ & = \frac{1}{2} \sum_{l_\mathrm{BU}} \sum_{l_\mathrm{R}} \left | [\mathbf{P}]_{l_\mathrm{BU},l_\mathrm{R}} \right | \frac{ w_{\mathrm{R}}(\psi_{l_\mathrm{R}})^{\frac{1}{2}} }{ w_{\mathrm{BU}}([\boldsymbol{\theta}]_{l_\mathrm{BU}},[\boldsymbol{\phi}]_{l_\mathrm{BU}})^{\frac{1}{2}} } \mathbf{a}^{H}_{N_\mathrm{R}} ( \psi_{l_\mathrm{R}} ) \mathbf{W}_{\mathrm{R}} \mathbf{a}_{N_\mathrm{R}} ( \psi_{l_\mathrm{R}} ) \\ & \quad + \frac{1}{2} \sum_{l_\mathrm{BU}} \sum_{l_\mathrm{R}} \left | [\mathbf{P}]_{l_\mathrm{BU},l_\mathrm{R}} \right | \frac{ w_{\mathrm{BU}}([\boldsymbol{\theta}]_{l_\mathrm{BU}},[\boldsymbol{\phi}]_{l_\mathrm{BU}})^{\frac{1}{2}} }{ w_{\mathrm{R}}(\psi_{l_\mathrm{R}})^{\frac{1}{2}} } \\ & \qquad \cdot \big[ \mathbf{a}_{N_\mathrm{B}} ( [\boldsymbol{\theta}]_{l_\mathrm{BU}} ) \otimes \mathbf{a}_{N_\mathrm{U}} ( [\boldsymbol{\phi}]_{l_\mathrm{BU}} ) \big]^{H} \mathbf{W}_{\mathrm{BU}} \big[ \mathbf{a}_{N_\mathrm{B}} ( [\boldsymbol{\theta}]_{l_\mathrm{BU}} ) \otimes \mathbf{a}_{N_\mathrm{U}} ( [\boldsymbol{\phi}]_{l_\mathrm{BU}} ) \big] \\ & = \sum_{l_\mathrm{BU}} \sum_{l_\mathrm{R}} \left | [\mathbf{P}]_{l_\mathrm{BU},l_\mathrm{R}} \right | w_{\mathrm{BU}}([\boldsymbol{\theta}]_{l_\mathrm{BU}},[\boldsymbol{\phi}]_{l_\mathrm{BU}})^{-\frac{1}{2}} w_{\mathrm{R}}(\psi_{l_\mathrm{R}})^{-\frac{1}{2}}.
		\end{aligned}
	\end{equation}
	Since \eqref{SDPWleq} holds for any partially decoupled atomic decomposition of $\mathbf{H}$, we have that $\mathrm{SDP}_{\mathbf{W}_{\mathrm{R}},\mathbf{W}_{\mathrm{BU}}}(\mathbf{H}) \leq \| \mathbf{H} \|_{\mathcal{A}^{w_{\mathrm{R}},w_{\mathrm{BU}}}}$.
	
	We next show that $\mathrm{SDP}_{\mathbf{W}_{\mathrm{R}},\mathbf{W}_{\mathrm{BU}}}(\mathbf{H}) = \| \mathbf{H} \|_{\mathcal{A}^{w_{\mathrm{R}},w_{\mathrm{BU}}}}$ under the three conditions in Theorem \ref{Theorem1}. Without loss of generality, we assume that there is exactly one non-zero element in each row of $\mathbf{C}$, denoted as $[\mathbf{C}]_{i,k_{i}}$ in the $i$-th row. By derivations similar to those of Theorem \ref{Theorem1}, we obtain that \eqref{HAtoDec*}, \eqref{diagsucceq1} and \eqref{diagsucceq2} hold. Then, it follows that
	\begin{equation} \label{SDPWgeq1}
		\begin{aligned}
			\mathrm{SDP}_{\mathbf{W}_{\mathrm{R}},\mathbf{W}_{\mathrm{BU}}}(\mathbf{H}) & = \frac{1}{2} \mathrm{tr} \big(\mathbf{W}_{\mathrm{R}}\mathcal{T}_{N_\mathrm{R}}(\mathbf{t}^{*})\big) + \frac{1}{2} \mathrm{tr} \big(\mathbf{W}_{\mathrm{BU}}\mathcal{T}_{[N_\mathrm{B},N_\mathrm{U}]}(\mathbf{T}^{*})\big) \\ & = \frac{1}{2} \sum_{l_\mathrm{R}=1}^{L_\mathrm{R}^{*}} [\mathbf{p}_\mathrm{R}^{*}]_{l_\mathrm{R}} \mathbf{a}^{H}_{N_\mathrm{R}} ( [\boldsymbol{\psi}_\mathrm{R}^{*}]_{l_{R}} ) \mathbf{W}_{\mathrm{R}} \mathbf{a}_{N_\mathrm{R}} ( [\boldsymbol{\psi}_\mathrm{R}^{*}]_{l_{R}} ) + \frac{1}{2} \sum_{l_\mathrm{BU}=1}^{L_\mathrm{BU}^{*}} [\mathbf{p}_\mathrm{BU}^{*}]_{l_\mathrm{BU}} \\ & \quad \cdot \big[ \mathbf{a}_{N_\mathrm{B}} ( [\boldsymbol{\theta}_\mathrm{B}^{*}]_{l_{BU}} ) \otimes \mathbf{a}_{N_\mathrm{U}} ( [\boldsymbol{\phi}_\mathrm{U}^{*}]_{l_{BU}} ) \big]^{H} \mathbf{W}_{\mathrm{BU}} \big[ \mathbf{a}_{N_\mathrm{B}} ( [\boldsymbol{\theta}_\mathrm{B}^{*}]_{l_{BU}} ) \otimes \mathbf{a}_{N_\mathrm{U}} ( [\boldsymbol{\phi}_\mathrm{U}^{*}]_{l_{BU}} ) \big] \\ & = \frac{1}{2} \sum_{l_\mathrm{R}=1}^{L_\mathrm{R}^{*}} w_{\mathrm{R}}([\boldsymbol{\psi}_\mathrm{R}^{*}]_{l_{R}})^{-1} [\mathbf{p}_\mathrm{R}^{*}]_{l_\mathrm{R}} + \frac{1}{2} \sum_{l_\mathrm{BU}=1}^{L_\mathrm{BU}^{*}} w_{\mathrm{BU}}([\boldsymbol{\theta}_\mathrm{B}^{*}]_{l_{BU}},[\boldsymbol{\phi}_\mathrm{U}^{*}]_{l_{BU}})^{-1} [\mathbf{p}_\mathrm{BU}^{*}]_{l_\mathrm{BU}} \\ & \geq \frac{1}{2} \sum_{l_\mathrm{R}=1}^{L_\mathrm{R}^{*}} \sum_{l_\mathrm{BU}=1}^{L_\mathrm{BU}^{*}} w_{\mathrm{R}}([\boldsymbol{\psi}_\mathrm{R}^{*}]_{l_{R}})^{-1} \frac{\left| [\mathbf{C}]_{l_\mathrm{BU},l_\mathrm{R}} \right|^{2}}{[\mathbf{p}_\mathrm{BU}^{*}]_{l_\mathrm{BU}}} + \frac{1}{2} \sum_{l_\mathrm{BU}=1}^{L_\mathrm{BU}^{*}} w_{\mathrm{BU}}([\boldsymbol{\theta}_\mathrm{B}^{*}]_{l_{BU}},[\boldsymbol{\phi}_\mathrm{U}^{*}]_{l_{BU}})^{-1} [\mathbf{p}_\mathrm{BU}^{*}]_{l_\mathrm{BU}} \\ & \geq \sum_{l_\mathrm{BU}=1}^{L_\mathrm{BU}^{*}} w_{\mathrm{R}}([\boldsymbol{\psi}_\mathrm{R}^{*}]_{k_{l_\mathrm{BU}}})^{-\frac{1}{2}} w_{\mathrm{BU}}([\boldsymbol{\theta}_\mathrm{B}^{*}]_{l_{BU}},[\boldsymbol{\phi}_\mathrm{U}^{*}]_{l_{BU}})^{-\frac{1}{2}} \left| [\mathbf{C}]_{l_\mathrm{BU},k_{l_\mathrm{BU}}} \right| \geq \| \mathbf{H} \|_{\mathcal{A}^{w_{\mathrm{R}},w_{\mathrm{BU}}}}.
		\end{aligned}
	\end{equation}
	Combining \eqref{SDPWgeq1} and \eqref{SDPWleq}, we obtain that $\mathrm{SDP}(\mathbf{H}) = \| \mathbf{H} \|_{\mathcal{A}^{w_{\mathrm{R}},w_{\mathrm{BU}}}}$. Similarly, if there is exactly one non-zero element in each column of $\mathbf{C}$, denoted as $[\mathbf{C}]_{k_{j},j}$ in the $j$-th column, we have \eqref{HAtoDec*}, \eqref{diagsucceq3} and \eqref{diagsucceq4}. It follows that
	\begin{equation} \label{SDPWgeq2}
		\begin{aligned}
			\mathrm{SDP}_{\mathbf{W}_{\mathrm{R}},\mathbf{W}_{\mathrm{BU}}}(\mathbf{H}) = & \frac{1}{2} \sum_{l_\mathrm{R}=1}^{L_\mathrm{R}^{*}} w_{\mathrm{R}}([\boldsymbol{\psi}_\mathrm{R}^{*}]_{l_{R}})^{-1} [\mathbf{p}_\mathrm{R}^{*}]_{l_\mathrm{R}} + \frac{1}{2} \sum_{l_\mathrm{BU}=1}^{L_\mathrm{BU}^{*}} w_{\mathrm{BU}}([\boldsymbol{\theta}_\mathrm{B}^{*}]_{l_{BU}},[\boldsymbol{\phi}_\mathrm{U}^{*}]_{l_{BU}})^{-1} [\mathbf{p}_\mathrm{BU}^{*}]_{l_\mathrm{BU}} \\ \geq & \frac{1}{2} \sum_{l_\mathrm{R}=1}^{L_\mathrm{R}^{*}} w_{\mathrm{R}}([\boldsymbol{\psi}_\mathrm{R}^{*}]_{l_{R}})^{-1} [\mathbf{p}_\mathrm{R}^{*}]_{l_\mathrm{R}} + \frac{1}{2} \sum_{l_\mathrm{BU}=1}^{L_\mathrm{BU}^{*}} \sum_{l_\mathrm{R}=1}^{L_\mathrm{R}^{*}} w_{\mathrm{BU}}([\boldsymbol{\theta}_\mathrm{B}^{*}]_{l_{BU}},[\boldsymbol{\phi}_\mathrm{U}^{*}]_{l_{BU}})^{-1} \frac{\left| [\mathbf{C}]_{l_\mathrm{BU},l_\mathrm{R}} \right|^{2}}{[\mathbf{p}_\mathrm{R}^{*}]_{l_\mathrm{R}}} \\ \geq & \sum_{l_\mathrm{R}=1}^{L_\mathrm{R}^{*}} w_{\mathrm{R}}([\boldsymbol{\psi}_\mathrm{R}^{*}]_{l_{R}})^{-\frac{1}{2}} w_{\mathrm{BU}}([\boldsymbol{\theta}_\mathrm{B}^{*}]_{k_{l_\mathrm{R}}},[\boldsymbol{\phi}_\mathrm{U}^{*}]_{k_{l_\mathrm{R}}})^{-\frac{1}{2}} \left| [\mathbf{C}]_{k_{l_\mathrm{R}},l_\mathrm{R}} \right| \geq \| \mathbf{H} \|_{\mathcal{A}^{w_{\mathrm{R}},w_{\mathrm{BU}}}}.
		\end{aligned}
	\end{equation}
	Combining \eqref{SDPWgeq2} and \eqref{SDPWleq}, we have $\mathrm{SDP}(\mathbf{H}) = \| \mathbf{H} \|_{\mathcal{A}^{w_{\mathrm{R}},w_{\mathrm{BU}}}}$ and complete the proof.}

%\section*{Acknowledgment}
%This should be a simple paragraph before the References to thank those individuals and institutions who have supported your work on this article.

\bibliographystyle{IEEEtran}
\bibliography{IEEEabrv,ref_RPDANM-APC}

\begin{thebibliography}{10}
\providecommand{\url}[1]{#1}
\csname url@samestyle\endcsname
\providecommand{\newblock}{\relax}
\providecommand{\bibinfo}[2]{#2}
\providecommand{\BIBentrySTDinterwordspacing}{\spaceskip=0pt\relax}
\providecommand{\BIBentryALTinterwordstretchfactor}{4}
\providecommand{\BIBentryALTinterwordspacing}{\spaceskip=\fontdimen2\font plus
\BIBentryALTinterwordstretchfactor\fontdimen3\font minus
  \fontdimen4\font\relax}
\providecommand{\BIBforeignlanguage}[2]{{%
\expandafter\ifx\csname l@#1\endcsname\relax
\typeout{** WARNING: IEEEtran.bst: No hyphenation pattern has been}%
\typeout{** loaded for the language `#1'. Using the pattern for}%
\typeout{** the default language instead.}%
\else
\language=\csname l@#1\endcsname
\fi
#2}}
\providecommand{\BIBdecl}{\relax}
\BIBdecl

\bibitem{chu2023GLOBECOM}
Y.~Chu, Z.~Wei, Z.~Yang, and D.~W.~K. Ng, ``Channel estimation for {RIS}-aided
  {MIMO} systems via partially decoupled atomic norm minimization,'' in
  \emph{Proc. IEEE Global Commun. Conf. (GLOBECOM)}, 2023, pp. 6615--6620.

\bibitem{wu2019towards}
Q.~Wu and R.~Zhang, ``Towards smart and reconfigurable environment: Intelligent
  reflecting surface aided wireless network,'' \emph{{IEEE} Commun. Mag.},
  vol.~58, no.~1, pp. 106--112, Jan. 2020.

\bibitem{yuan2021reconfigurable}
X.~Yuan, Y.-J.~A. Zhang, Y.~Shi, W.~Yan, and H.~Liu,
  ``Reconfigurable-intelligent-surface empowered wireless communications:
  Challenges and opportunities,'' \emph{{IEEE} Wireless Commun.}, vol.~28,
  no.~2, pp. 136--143, Apr. 2021.

\bibitem{liu2021reconfigurable}
Y.~Liu, X.~Liu, X.~Mu, T.~Hou, J.~Xu, M.~Di~Renzo, and N.~Al-Dhahir,
  ``Reconfigurable intelligent surfaces: Principles and opportunities,''
  \emph{{IEEE} Commun. Surveys Tuts.}, vol.~23, no.~3, pp. 1546--1577, Aug.
  2021.

\bibitem{huang2019reconfigurable}
C.~Huang, A.~Zappone, G.~C. Alexandropoulos, M.~Debbah, and C.~Yuen,
  ``Reconfigurable intelligent surfaces for energy efficiency in wireless
  communication,'' \emph{{IEEE} Trans. Wireless Commun.}, vol.~18, no.~8, pp.
  4157--4170, Aug. 2019.

\bibitem{di2020reconfigurable}
M.~Di~Renzo, K.~Ntontin, J.~Song, F.~H. Danufane, X.~Qian, F.~Lazarakis,
  J.~De~Rosny, D.-T. Phan-Huy, O.~Simeone, R.~Zhang, M.~Debbah, G.~Lerosey,
  M.~Fink, S.~Tretyakov, and S.~Shamai, ``Reconfigurable intelligent surfaces
  vs. relaying: Differences, similarities, and performance comparison,''
  \emph{IEEE Open J. Commun. Soc.}, vol.~1, pp. 798--807, Jul. 2020.

\bibitem{basar2019wireless}
E.~Basar, M.~Di~Renzo, J.~De~Rosny, M.~Debbah, M.-S. Alouini, and R.~Zhang,
  ``Wireless communications through reconfigurable intelligent surfaces,''
  \emph{{IEEE} Access}, vol.~7, pp. 116\,753--116\,773, Sep. 2019.

\bibitem{wei2021sum}
Z.~Wei, Y.~Cai, Z.~Sun, D.~W.~K. Ng, J.~Yuan, M.~Zhou, and L.~Sun, ``Sum-rate
  maximization for {IRS}-assisted {UAV} {OFDMA} communication systems,''
  \emph{{IEEE} Trans. Wireless Commun.}, vol.~20, no.~4, pp. 2530--2550, Apr.
  2021.

\bibitem{hu2021robust}
S.~Hu, Z.~Wei, Y.~Cai, C.~Liu, D.~W.~K. Ng, and J.~Yuan, ``Robust and secure
  sum-rate maximization for multiuser {MISO} downlink systems with
  self-sustainable {IRS},'' \emph{{IEEE} Trans. Commun.}, vol.~69, no.~10, pp.
  7032--7049, Oct. 2021.

\bibitem{wu2021intelligent}
Q.~Wu, S.~Zhang, B.~Zheng, C.~You, and R.~Zhang, ``Intelligent reflecting
  surface-aided wireless communications: A tutorial,'' \emph{{IEEE} Trans.
  Commun.}, vol.~69, no.~5, pp. 3313--3351, May 2021.

\bibitem{jensen2020optimal}
T.~L. Jensen and E.~De~Carvalho, ``An optimal channel estimation scheme for
  intelligent reflecting surfaces based on a minimum variance unbiased
  estimator,'' in \emph{Proc. IEEE Intern. Conf. on Acoust., Speech and Signal
  Process.}, 2020, pp. 5000--5004.

\bibitem{wei2020parallel}
L.~Wei, C.~Huang, G.~C. Alexandropoulos, and C.~Yuen, ``Parallel factor
  decomposition channel estimation in {RIS}-assisted multi-user {MISO}
  communication,'' in \emph{IEEE 11th Sensor Array and Multichannel Signal
  Processing Workshop (SAM)}, 2020, pp. 1--5.

\bibitem{de2021channel}
G.~T. de~Ara{\'u}jo, A.~L. de~Almeida, and R.~Boyer, ``Channel estimation for
  intelligent reflecting surface assisted {MIMO} systems: A tensor modeling
  approach,'' \emph{{IEEE} J. Sel. Topics Signal Process.}, vol.~15, no.~3, pp.
  789--802, Apr. 2021.

\bibitem{wang2020compressed}
P.~Wang, J.~Fang, H.~Duan, and H.~Li, ``Compressed channel estimation for
  intelligent reflecting surface-assisted millimeter wave systems,''
  \emph{{IEEE} Signal Process. Lett.}, vol.~27, pp. 905--909, Jun. 2020.

\bibitem{ardah2021trice}
K.~Ardah, S.~Gherekhloo, A.~L. de~Almeida, and M.~Haardt, ``{TRICE}: A channel
  estimation framework for {RIS}-aided millimeter-wave {MIMO} systems,''
  \emph{{IEEE} Signal Process. Lett.}, vol.~28, pp. 513--517, Mar. 2021.

\bibitem{he2021channel}
J.~He, H.~Wymeersch, and M.~Juntti, ``Channel estimation for {RIS}-aided
  mm{W}ave {MIMO} systems via atomic norm minimization,'' \emph{{IEEE} Trans.
  Wireless Commun.}, vol.~20, no.~9, pp. 5786--5797, Sep. 2021.

\bibitem{he2021leveraging}
J.~\vspace{0mm}He, H.~Wymeersch, and M.~Juntti, ``Leveraging location
  information for {RIS}-aided mm{W}ave {MIMO} communications,'' \emph{{IEEE}
  Wireless Commun. Lett.}, vol.~10, no.~7, pp. 1380--1384, Jul. 2021.

\bibitem{chung2021location}
H.~Chung and S.~Kim, ``Location-aware beam training and multi-dimensional
  {ANM}-based channel estimation for {RIS}-aided mm{W}ave systems,''
  \emph{{IEEE} Trans. Wireless Commun.}, vol.~23, no.~1, pp. 652--666, Jan.
  2024.

\bibitem{pati1993orthogonal}
Y.~C. Pati, R.~Rezaiifar, and P.~S. Krishnaprasad, ``Orthogonal matching
  pursuit: Recursive function approximation with applications to wavelet
  decomposition,'' in \emph{Proceedings of 27th Asilomar Conference on Signals,
  Systems and Computers}, 1993, pp. 40--44.

\bibitem{stoica2012sparse}
P.~Stoica and P.~Babu, ``Sparse estimation of spectral lines: Grid selection
  problems and their solutions,'' \emph{{IEEE} Trans. Signal Process.},
  vol.~60, no.~2, pp. 962--967, Feb. 2012.

\bibitem{yang2018sparse}
Z.~Yang, J.~Li, P.~Stoica, and L.~Xie, ``Sparse methods for
  direction-of-arrival estimation,'' in \emph{Academic Press Library in Signal
  Processing, Volume 7}.\hskip 1em plus 0.5em minus 0.4em\relax Elsevier, 2018,
  pp. 509--581.

\bibitem{yang2016vandermonde}
Z.~Yang, L.~Xie, and P.~Stoica, ``Vandermonde decomposition of multilevel
  {T}oeplitz matrices with application to multidimensional super-resolution,''
  \emph{{IEEE} Trans. Inf. Theory}, vol.~62, no.~6, pp. 3685--3701, Jun. 2016.

\bibitem{meyr1998digital}
H.~Meyr, M.~Moeneclaey, and S.~A. Fechtel, \emph{Digital Communication
  Receivers: Synchronization, Channel Estimation and Signal Processing}.\hskip
  1em plus 0.5em minus 0.4em\relax Wiley, 1998.

\bibitem{pun2006maximum}
M.-O. Pun, M.~Morelli, and C.-C. Kuo, ``Maximum-likelihood synchronization and
  channel estimation for {OFDMA} uplink transmissions,'' \emph{{IEEE} Trans.
  Commun.}, vol.~54, no.~4, pp. 726--736, Apr. 2006.

\bibitem{saleh1987statistical}
A.~Saleh and R.~Valenzuela, ``A statistical model for indoor multipath
  propagation,'' \emph{{IEEE} J. Sel. Areas Commun.}, vol.~5, no.~2, pp.
  128--137, Feb. 1987.

\bibitem{van2004optimum}
H.~L. Van~Trees, \emph{{Optimum Array Processing: Part {IV} of Detection,
  Estimation and Modulation Theory}}.\hskip 1em plus 0.5em minus 0.4em\relax
  John Wiley \& Sons, 2004.

\bibitem{NairOCEANS}
B.~M. Nair, R.~J. P, A.~Kumar, and R.~Bahl, ``Adaptive beamformer based
  left-right ambiguity resolution using twin array,'' in \emph{OCEANS}, 2022,
  pp. 1--8.

\bibitem{zhang2020capacity}
S.~Zhang and R.~Zhang, ``Capacity characterization for intelligent reflecting
  surface aided {MIMO} communication,'' \emph{{IEEE} J. Sel. Areas Commun.},
  vol.~38, no.~8, pp. 1823--1838, Aug. 2020.

\bibitem{pan2020multicell}
C.~Pan, H.~Ren, K.~Wang, W.~Xu, M.~Elkashlan, A.~Nallanathan, and L.~Hanzo,
  ``Multicell {MIMO} communications relying on intelligent reflecting
  surfaces,'' \emph{{IEEE} Trans. Wireless Commun.}, vol.~19, no.~8, pp.
  5218--5233, Aug. 2020.

\bibitem{rao1998matrix}
C.~R. Rao and M.~B. Rao, \emph{Matrix Algebra and its Applications to
  Statistics and Econometrics}.\hskip 1em plus 0.5em minus 0.4em\relax World
  Scientific, 1998.

\bibitem{chandrasekaran2012convex}
V.~Chandrasekaran, B.~Recht, P.~A. Parrilo, and A.~S. Willsky, ``The convex
  geometry of linear inverse problems,'' \emph{Foundations of Computational
  mathematics}, vol.~12, pp. 805--849, Oct. 2012.

\bibitem{chi2014compressive}
Y.~Chi and Y.~Chen, ``Compressive two-dimensional harmonic retrieval via atomic
  norm minimization,'' \emph{{IEEE} Trans. Signal Process.}, vol.~63, no.~4,
  pp. 1030--1042, Feb. 2015.

\bibitem{zhang2019efficient}
Z.~Zhang, Y.~Wang, and Z.~Tian, ``Efficient two-dimensional line spectrum
  estimation based on decoupled atomic norm minimization,'' \emph{Signal
  Process.}, vol. 163, pp. 95--106, 2019.

\bibitem{xi2017super}
F.~Xi, S.~Chen, and Z.~Liu, ``Super-resolution delay-{D}oppler estimation for
  sub-{N}yquist radar via atomic norm minimization,'' in \emph{Proc. IEEE
  Intern. Conf. on Acoust., Speech and Signal Process.}, 2017, pp. 4326--4330.

\bibitem{yang2016exact}
Z.~Yang and L.~Xie, ``Exact joint sparse frequency recovery via optimization
  methods,'' \emph{{IEEE} Trans. Signal Process.}, vol.~64, no.~19, pp.
  5145--5157, Oct. 2016.

\bibitem{toh1999sdpt3}
K.-C. Toh, M.~J. Todd, and R.~H. T{\"u}t{\"u}nc{\"u}, ``{SDPT}3—a {MATLAB}
  software package for semidefinite programming, version 1.3,''
  \emph{Optimization Methods and Software}, vol.~11, no. 1-4, pp. 545--581,
  1999.

\bibitem{vandenberghe1996semidefinite}
L.~Vandenberghe and S.~Boyd, ``Semidefinite programming,'' \emph{SIAM review},
  vol.~38, no.~1, pp. 49--95, 1996.

\bibitem{yang2015enhancing}
Z.~Yang and L.~Xie, ``Enhancing sparsity and resolution via reweighted atomic
  norm minimization,'' \emph{{IEEE} Trans. Signal Process.}, vol.~64, no.~4,
  pp. 995--1006, Feb. 2016.

\bibitem{chu2023new}
Y.~Chu, Z.~Wei, and Z.~Yang, ``New reweighted atomic norm minimization approach
  for line spectral estimation,'' \emph{Signal Process.}, vol. 206, p. 108897,
  2023.

\bibitem{rao1989performance}
B.~D. Rao and K.~S. Hari, ``Performance analysis of root-{MUSIC},''
  \emph{{IEEE} Trans. Acoust., Speech, Signal Process.}, vol.~37, no.~12, pp.
  1939--1949, Dec. 1989.

\bibitem{heckel2017generalized}
R.~Heckel and M.~Soltanolkotabi, ``Generalized line spectral estimation via
  convex optimization,'' \emph{{IEEE} Trans. Inf. Theory}, vol.~64, no.~6, pp.
  4001--4023, Jun. 2018.

\bibitem{horn2012matrix}
R.~A. Horn and C.~R. Johnson, \emph{Matrix Analysis}.\hskip 1em plus 0.5em
  minus 0.4em\relax Cambridge university press, 2012.

\bibitem{zhang2006schur}
F.~Zhang, \emph{The Schur Complement and Its Applications}.\hskip 1em plus
  0.5em minus 0.4em\relax Springer Science \& Business Media, 2006, vol.~4.

\end{thebibliography}

\end{document}